%% file: PRL_arxiv_v2.tex
\documentclass[aps,english,prb,floatfix,amsmath,superscriptaddress,tightenlines,twocolumn]{revtex4-2}
\usepackage{mathtools, amssymb}
\usepackage{graphicx}
\usepackage{tikz}
\usepackage{tikz-3dplot}
\usetikzlibrary{spy}
\usetikzlibrary{arrows.meta}
\usetikzlibrary{calc}
\usetikzlibrary{decorations.pathreplacing,calligraphy}
\usepackage{intcalc}
\usepackage{physics}
\usepackage[caption=false]{subfig}
\usepackage[shortlabels]{enumitem}
\usepackage{soul}
\usepackage[utf8]{inputenc}
\usepackage{xcolor}
\usepackage{amsthm}
\usepackage{float}
\usepackage{comment}
\usepackage{enumitem}
\usepackage{tcolorbox}

\usepackage{thm-restate}

\newtheorem{lemma}{Lemma}
\newtheorem{definition}{Definition}
\newtheorem{proposition}{Proposition}

\newcommand{\QQ}[0]{\widehat{Q}}

\definecolor{BSorange}{RGB}{140,50,0}
\makeatother    

\usepackage{hyperref}
\hypersetup{colorlinks=true,linkcolor=red,citecolor=magenta,urlcolor=blue}
\usepackage[all]{hypcap}

\begin{document}

\title{Universal lower bound on topological entanglement entropy}
	\author{Isaac H. Kim}
	\affiliation{Department of Computer Science, University of California, Davis, CA 95616, USA}
	\author{Michael Levin}
	\affiliation{Kadanoff Center for Theoretical Physics, University of Chicago, Chicago, Illinois 60637, USA}
	
	\author{Ting-Chun Lin}
	\affiliation{Department of Physics, University of California at San Diego, La Jolla, CA 92093, USA}
    \affiliation{Hon Hai Research Institute, Taipei, Taiwan}
	
	\author{Daniel Ranard}
	\affiliation{Center for Theoretical Physics, Massachusetts Institute of Technology, Cambridge, MA 02139}
	\author{Bowen Shi}
	\affiliation{Department of Physics, University of California at San Diego, La Jolla, CA 92093, USA}
	\date{\today}

	\begin{abstract}
	Entanglement entropies of two-dimensional gapped ground states are expected to satisfy an area law, with a constant correction term known as the topological entanglement entropy (TEE).  In many models, the TEE takes a universal value that characterizes the underlying topological phase.  However, the TEE is not truly universal: it can differ even for two states related by constant-depth circuits, which are necessarily in the same phase.  The difference between the TEE and the value predicted by the anyon theory is often called the \emph{spurious} topological entanglement entropy. We show that this spurious contribution is always nonnegative, thus the value predicted by the anyon theory provides a universal lower bound.
	This observation also leads to a definition of TEE which is invariant under constant-depth quantum circuits. 
	

	\end{abstract}
	\maketitle

Ground states of 2D gapped Hamiltonians are believed to satisfy an area law: the entanglement entropy of a region cannot increase faster than its perimeter.  In many examples, the entropy of the reduced density matrix on a disk $A$ takes the form 
\begin{equation}
    S(\sigma_A) = \alpha  |\partial A|  - \gamma + \cdots \label{eq:area-law}
\end{equation}
where $\alpha |\partial A|$ is the leading ``area law'' term proportional to the boundary length, $\gamma$ is a constant term, and the ellipsis represents terms that vanish for large regions.

 The constant term $\gamma$, under natural assumptions, was argued to be universal, i.e.\ the same for all gapped ground states in a given phase~\cite{Kitaev2006,Levin2006}. In particular, $\gamma$ takes a form determined solely by the underlying anyon theory of the phase, $\gamma=\log \mathcal{D}$,
 where $\mathcal{D}= \sqrt{\sum_a d_a^2}$ is the total quantum dimension of the anyons and $d_a$ is the quantum dimension of the anyon $a$. Given its connection to anyons, the constant $\gamma$ has been termed the ``topological entanglement entropy'' (TEE).
The TEE can be computed in both non-solvable~\cite{Isakov2011,Jiang2012} and solvable models~\cite{hamma2005bipartite,Kitaev1997,Levin2005,Hahn2020,lin2021generalized}, and it is often used as a smoking gun signature of topological order, or to distinguish two phases.

A common way to extract the TEE is to use a judicious linear combination of entropies of adjacent regions. We focus on the definition in Ref.~\cite{Levin2006}, where the TEE $\gamma$ is defined using the conditional mutual information $I(A:C|B)_{\sigma}:= S(\sigma_{AB}) + S(\sigma_{BC}) - S(\sigma_B) - S(\sigma_{ABC})$ of regions $A,B,$ and $C$ forming an annulus as in Fig.~\ref{fig:special_circuit}(a):
 \begin{align} \label{eq:gamma_CMI}
     I(A:C|B)_{\sigma} \equiv 2\gamma,
 \end{align}
 where $\sigma$ is a ground state.

 While the TEE is a useful diagnostic of topological order, it was soon observed \cite{BravyiUnpublished} that it is not a genuine invariant of the topological phase, unlike e.g.\ \cite{haah2016invariant,kato2020entropic}.  Two ground states are in the same phase if they are connected by a constant-depth circuit consisting of local gates. But $\gamma$ as defined by Eq.~\eqref{eq:gamma_CMI} can change
 under such a circuit. In fact, a shallow circuit acting on a product state may achieve a  nonzero value of $I(A:C|B)$ for arbitrarily large regions \cite{cano2015interactions, Zou2016,Williamson2019}.  Deviations of $\gamma$ from the purportedly universal value $\log \mathcal{D}$ have been called ``spurious'' contributions, or the spurious TEE. States with spurious TEE exist in both trivial and non-trivial topological phases.  These examples often arise from symmetry-protected topological phases (SPTs) \cite{Zou2016, fliss2017interface,santos2018symmetry, Williamson2019,stephen2019detecting} but perhaps not always \cite{Kato2020}.  

 Our main result partially restores the universality of the TEE by showing the spurious contribution is always nonnegative.  Thus $\log \mathcal{D}$ provides a universal \textit{lower bound} for the TEE $\gamma$:
 \begin{align} \label{eq:gamma_nonneg}
      \gamma \geq \log \mathcal{D} .
 \end{align}
 This observation leads directly to a definition of TEE on infinite systems which is invariant under constant-depth circuits, by minimizing the ordinary TEE over such circuits. More specifically, for a state $\rho$ defined on the infinite 2D plane, the following quantity 
 \begin{align} \label{eq:gamma_min}
     \gamma_{\textrm{min}} &=  \lim_{R \to \infty} \min_{U} \frac{1}{2} I(A_R:C_R|B_R)_{U\rho U^\dagger} \ , 
 \end{align}
 yields $\log \mathcal{D}$ (for a class of states elaborated below), where the regions $A_R,B_R,C_R$ have a radius and thickness of order $R$, and the minimum is taken over circuits $U$ of depth $d<c R$ for some fixed constant $c \in (0,1)$. 
 
  Our result in Eq.~\eqref{eq:gamma_nonneg} helps restore the TEE as a rigorous diagnostic to distinguish topological phases, albeit with limitations.  For instance, if a state has $I(A:C|B)= \log 2$ for some large regions, it may still be in the trivial phase (where $2 \log \mathcal{D} =0$), and indeed such examples exist. But it \textit{cannot} be in the same phase as the toric code, which has $2 \log \mathcal{D} = 2 \log 2$; the latter would require a negative spurious TEE, which we rule out.  
 
\textbf{\textit{Setup ---}}
We now explain our main result more precisely.
We consider a special class of bosonic quantum many-body states, defined on infinite two-dimensional lattices, 
which we refer to as ``reference states,''closely related to the states considered in Ref.~\cite{Shi2020}.
\begin{definition}
 \label{assumption:TEE-inv}
  A state $\sigma$ is a reference state if (i) the TEE calculated as $\gamma_0 = \frac{1}{2}I(A:C|B)_{\sigma}$ is the same for any choice of regions topologically equivalent to Fig.~\ref{fig:special_circuit}(a) and (ii) the mutual information between two subsystems is zero for any two non-adjacent subsystems. 
\end{definition}
\noindent

Our main technical result is the following inequality, which holds for any reference state $\sigma$, and for any circuit $U$ whose depth is small compared to the radius and thickness of the annulus $ABC$:
 \begin{align} \label{eq:rho_sigma_CMI} 
      I(A:C|B)_{U\sigma U^{\dagger}} \geq I(A:C|B)_\sigma \equiv 2 \gamma_0.
 \end{align}
In other words, we show that constant-depth circuits can never decrease the TEE, \textit{when acting on a reference state.} 

To understand the implications of this result, note that the set of reference states includes all ground states of string-net~\cite{Levin2005,Hahn2020,lin2021generalized} and quantum double models~\cite{Kitaev1997}, and more generally any state satisfying the entanglement bootstrap axioms~\cite{Shi2020}. 
For all of these examples, the RHS of Eq.~\eqref{eq:rho_sigma_CMI} is known to equal $2\log \mathcal{D}$~\cite{Levin2006,Kitaev2006,Shi2020}. Therefore, \eqref{eq:rho_sigma_CMI} implies the claimed lower bound \eqref{eq:gamma_nonneg} for any state obtained by a constant-depth circuit acting on a string-net, quantum double, or entanglement bootstrap state. 
(We discuss a generalization of \eqref{eq:gamma_nonneg} to general 2D gapped ground states in Appendix \ref{sec:stacking}.)
Similarly, we deduce \eqref{eq:gamma_min} with $\gamma_{\text{min}} = \log \mathcal{D}$ for any state $\rho$ given by a finite-depth circuit $V$ applied to a reference state, where the minimum is achieved by the circuit $U = V^{-1}$.

While we work in the plane for concreteness, our proof also applies to the TEE defined on any disk-like region embedded in an arbitrary manifold.

\textbf{\textit{Example: Toric code ---}}
To explain the key idea behind our proof, it is instructive to first focus on a concrete reference state $\sigma$, namely the toric code ground state~\cite{Kitaev1997} on a plane. 
(Our argument here will rely on special properties of the toric code state, but later we will generalize the proof to all states satisfying Definition \ref{assumption:TEE-inv}.)
For this state, $I(A: C|B)_\sigma= 2 \log 2$ so that $\gamma_0 = \log 2$~\cite{hamma2005bipartite}.
If we now apply a constant-depth quantum circuit $U$, defining $\widetilde{\sigma} =U \sigma U^\dagger$, in general $I(A : C| B)_{\widetilde{\sigma}} \ne 2 \log 2$.
Nevertheless, we will show that for a sufficiently large annulus $ABC$, we still have the lower bound
\begin{equation}
I(A : C| B)_{\widetilde{\sigma}} \geq 2 \log 2.
\label{eq:tcbound}
\end{equation}

We first prove the bound (\ref{eq:tcbound}) for a special class of constant-depth circuits $U$, namely those that are supported within a constant distance of $BC$ [Fig.~\ref{fig:special_circuit}(b)]. Later we will extend this result to general constant-depth circuits.

\begin{figure}[t]
      \includegraphics[]{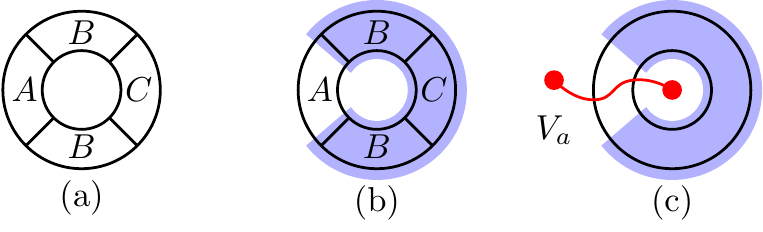}
 	\caption{(a) The partition used to calculate the TEE. (b) Support of the unitary $U$ considered. (c) String operator $V_a$ creates an anyon $a$ in the interior of the annulus and its antiparticle in the exterior.}  \label{fig:special_circuit}
 \end{figure}

Our basic strategy is to construct a state $\widetilde{\lambda}$ that is ``locally indistinguishable'' from $\widetilde{\sigma}$. More precisely, we will construct a state $\widetilde{\lambda}$ that is indistinguishable from $\widetilde{\sigma}$ over $AB$ and $BC$: that is, $\widetilde{\lambda}_{AB} = \widetilde{\sigma}_{AB}$ and $\widetilde{\lambda}_{BC} = \widetilde{\sigma}_{BC}$. We can then express $I(A: C|B)_{\widetilde{\sigma}}$ in terms of $I(A: C|B)_{\widetilde{\lambda}}$ using the identity
\begin{equation}
\label{eq:ISid}
    I(A: C|B)_{\widetilde{\sigma}} = I(A: C|B)_{\widetilde{\lambda}} +  S(\widetilde{\lambda}_{ABC}) - S(\widetilde{\sigma}_{ABC}). 
\end{equation}
By the strong subadditivity of the entropy (SSA)~\cite{Lieb1973}, $I(A: C|B)_{\widetilde{\lambda}} \geq 0$, so
\begin{equation}
I(A: C|B)_{\widetilde{\sigma}} \geq S(\widetilde{\lambda}_{ABC}) - S(\widetilde{\sigma}_{ABC}).
    \label{eq:gamma_lowerbound_ent_difference}
\end{equation}
We will obtain the desired lower bound (\ref{eq:tcbound}) from a judicious choice of $\widetilde{\lambda}$.

The easiest way to construct an appropriate $\widetilde{\lambda}$ is to first find a state $\lambda$ that is locally indistinguishable from the toric code ground state $\sigma$. More precisely, we need a $\lambda$ that is indistinguishable from $\sigma$ over the past light cone of $AB$ and $BC$ (with respect to $U$). Once we find such a $\lambda$, we can then set $\widetilde{\lambda} = U \lambda U^\dagger$. 

We construct such a $\lambda$ using a probabilistic mixture of toric code \emph{excited} states. (Later, we use a more general approach.)  For each anyon type $a\in \mathcal{C} = \{1, e, m, \epsilon \}$, we define a corresponding excited state $\rho^{(a)}$ by $\rho^{(a)}=V_a\sigma V_a^{\dagger}$, where $V_a$ is a unitary (open) string operator that places an anyon excitation $a$ in the interior of the annulus and its antiparticle in the exterior [Fig.~\ref{fig:special_circuit}(c)]. We then define $\lambda = \sum_a p_a \rho^{(a)}$ for some probability distribution $\{p_a: a\in \mathcal{C} \}$. Note that $\lambda$ has the requisite indistinguishability property as long as the endpoints of the string operators $V_a$ (where the anyons are created) are far enough away from the annulus to lie outside the past light cones of $AB$ and $BC$.

To proceed, we must evaluate the entropy difference $S(\widetilde{\lambda}_{ABC}) - S(\widetilde{\sigma}_{ABC})$. 
Here it is convenient to choose the path of the string operators $V_a$ so that they avoid the region of support of the constant-depth circuit $U$ (which by assumption is supported near $BC$). Then $V_a$ commutes with $U$ so $\widetilde{\lambda}$ can be written as a probabilistic mixture of the form
\begin{equation}
\widetilde{\lambda} = \sum_a p_a \widetilde{\rho}^{(a)}, \quad \quad \widetilde{\rho}^{(a)} = V_a \widetilde{\sigma} V_a^{\dagger}.
\end{equation}

Crucially, the $\widetilde{\rho}^{(a)}$ states have two simplifying properties: (i) different $\widetilde{\rho}^{(a)}_{ABC}$ are orthogonal, and (ii) $S(\widetilde{\rho}^{(a)}_{ABC}) = S(\widetilde{\sigma}_{ABC})$. Intuitively, property (i) follows from the fact that each $\widetilde{\rho}^{(a)}_{ABC}$ belongs to a different anyon sector on the annulus. More formally, (i) follows from the existence of a collection of (closed) string operators supported within $ABC$ that take on different eigenvalues in each state $\widetilde{\rho}^{(a)}_{ABC}$. (These string operators are simply the closed versions of $U V_a U^\dagger$; they can be drawn within $ABC$ whenever $ABC$ is wider than twice the circuit depth of $U$). Meanwhile, property (ii) follows from the fact that the $V_a$ are products of single-site unitaries; in particular, each $V_a$ can be written as a product of a unitary acting entirely within $ABC$ and a unitary acting entirely outside $ABC$, neither of which changes the entanglement entropy of $ABC$. 

Given properties (i) and (ii) of $\widetilde{\rho}^{(a)}$, the entropy difference can be computed as 
\begin{equation}
S(\widetilde{\lambda}_{ABC}) - S(\widetilde{\sigma}_{ABC}) = H(\{p_a \}),
\label{eq:tcent_diff}
\end{equation}
where $H(\{p_a \}) = - \sum_a p_a \log(p_a)$ is the Shannon entropy of the probability distribution $\{p_a\}$. Substituting (\ref{eq:tcent_diff}) into (\ref{eq:gamma_lowerbound_ent_difference}), we obtain
\begin{equation}
I(A: C|B)_{\widetilde{\sigma}} \geq H(\{p_a \}).
\label{eq:ent_lower_bound_toric}
\end{equation}
To get the best bound, we choose the probability distribution that maximizes $H(\{p_a \})$, namely the uniform $p_a = \frac{1}{4}$. Then $H(\{p_a \}) = 2 \log 2$, yielding the desired bound (\ref{eq:tcbound}).

\begin{figure}[t]
    \includegraphics[]{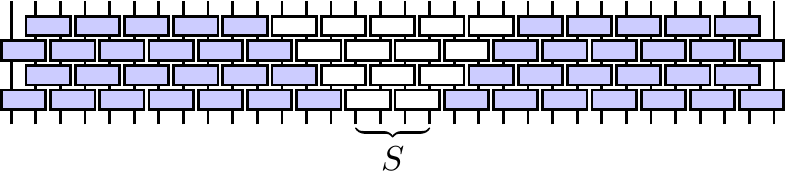}
	\caption{For any constant-depth circuit $U$, for any subsystem $S$, we can obtain a circuit $U'$ of same depth acting trivially on $S$, by removing from $U$ the ``light cone'' (white gates) of $S$.}\label{fig:past-causal-cone}
\end{figure}

To complete the argument, we extend the bound (\ref{eq:ent_lower_bound_toric}) to general constant-depth circuits $U$. First, recall the entanglement entropy of a subsystem is invariant under unitaries acting exclusively within the subsystem or its complement. Thus we can make the replacement
\begin{equation}
I(A: C|B)_{U \sigma U^\dagger} = I(A: C|B)_{U'{\sigma}U'^\dagger},
\label{eq:UtoUp}
\end{equation}
where $U'$ is a constant-depth quantum circuit that acts trivially deep in the interior of $A$ and also trivially far outside $ABC$ [Fig.~\ref{fig:deformation-A}(a)]. Here, we are using the fact that $U$ is a constant-depth quantum circuit, and therefore we can ``cancel out'' its action in a subsystem $S$ by multiplying by an appropriate unitary supported in the light cone of $S$ (Fig.~\ref{fig:past-causal-cone}). 

By SSA, $I(A: C|B)$ cannot increase when $A$ shrinks. Therefore
\begin{equation}
I(A: C|B)_{U'{\sigma}U'^\dagger} \geq I(A': C|B)_{U'{\sigma}U'^\dagger}    
\label{eq:ssa_atoap}
\end{equation}
where $A'\subset A$ is shown in Fig.~\ref{fig:deformation-A}(b). Finally, applying the same reasoning as in (\ref{eq:UtoUp}), we can replace 
\begin{equation}
I(A': C|B)_{U'{\sigma}U'^\dagger} = I(A': C|B)_{U''{\sigma}U''^\dagger}
\label{eq:UptoUp2}
\end{equation}
where $U''$ is a constant-depth circuit acting on the region shown in Fig.~\ref{fig:deformation-A}(c). Combining (\ref{eq:UtoUp}-\ref{eq:UptoUp2}), we deduce that 

\begin{equation}
I(A: C|B)_{U \sigma U^\dagger} \geq I(A': C|B)_{U''{\sigma}U''^\dagger}. 
\label{eq:ssa_atoap2}
\end{equation}
The lower bound (\ref{eq:ssa_atoap2}) is useful because it allows us to leverage our results from the first part of the proof. In particular, $U''$ is precisely the kind of constant-depth quantum circuit that we analyzed above, so $I(A': C|B)_{U''{\sigma}U''^\dagger} \geq 2 \log 2$ for any sufficiently large annulus $A'BC$. Substituting this inequality into (\ref{eq:ssa_atoap2}), we obtain the desired bound (\ref{eq:tcbound}).

\begin{figure}[t]
	\centering
	    \includegraphics[]{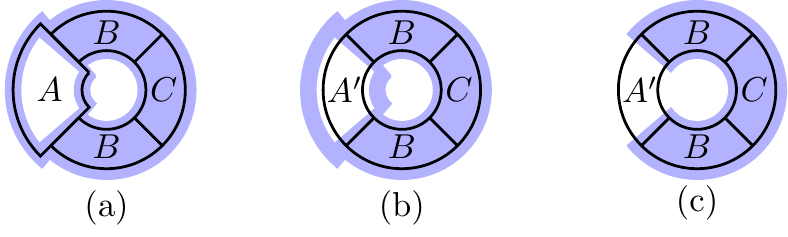}
	\caption{(a) We first remove gates from $U$ on a ``hole" within $A$; we call the new circuit $U'$. (b) We then deform $A$ to $A'\subset A$ such that the boundary of the annulus $A'BC$ is only partially covered. (c) We then further remove some of the gates in the vicinity of $A'$, obtaining $U''$. }
	\label{fig:deformation-A}
\end{figure}

\textbf{\textit{General case ---}}
Our proof for the toric code proceeded in three steps. First, we derived a lower bound (\ref{eq:gamma_lowerbound_ent_difference}) for $I(A:C|B)_{\widetilde{\sigma}}$ in terms of the entropy difference $S(\widetilde{\lambda}_{ABC}) - S(\widetilde{\sigma}_{ABC})$ where $\widetilde{\lambda}$ is any state that is indistinguishable from $\widetilde{\sigma}$ over $AB$ and $BC$. Second, we constructed an appropriate $\widetilde{\lambda}$ and computed the desired entropy difference (\ref{eq:tcent_diff}) in the special case where $U$ is a constant-depth circuit supported within a constant distance of $BC$ [Fig.~\ref{fig:special_circuit}(b)]. Combining these two results, we obtained the desired lower bound, but only for this special class of circuits $U$. In the third and final step, we extended this bound to arbitrary constant-depth circuits $U$ using the inequality (\ref{eq:ssa_atoap2}).

Conveniently, the first and third steps of our proof immediately generalize to any reference state $\sigma$ since they do not use any properties of the toric code. On the other hand, in the second step, we used the specific structure of the toric code string operators \footnote{In particular, we used the fact that the string operators of the toric code are depth-$1$ quantum circuits. In contrast, general anyon models can require higher depths, e.g., non-Abelian anyon strings need linear depth (\cite{ShiThesis}, Chapter 7).}, so we need a different argument for this step in the general case. In particular, instead of defining $\widetilde{\lambda}$ in terms of a mixture of excited states, we will now define it in terms of the \emph{maximum-entropy state} (``max-entropy state''). Consider a larger annulus $Y = ABC{\cup}\text{Supp}(U)$, with $U$ again as in Fig.~\ref{fig:special_circuit}(b). Define a density matrix $\lambda$ to be the maximum-entropy state consistent with the reduced density matrices of $\sigma$ over the past light cones of $AB$ and $BC$. We then define $\widetilde{\lambda} = U \lambda U^\dagger$.

By construction, $\widetilde{\lambda}$ is indistinguishable from $\widetilde{\sigma}$ over $AB$ and $BC$ and therefore the lower bound (\ref{eq:gamma_lowerbound_ent_difference}) still holds. The only remaining question is the value of the entropy difference on the right-hand side. We claim that
\begin{align}
    S(\widetilde{\lambda}_{ABC}) - S(\widetilde{\sigma}_{ABC}) & = S(\lambda_{Y}) - S(\sigma_{Y}) \label{eq:ent_diff_invariant} 
\end{align}
and in turn
\begin{align}
    S(\lambda_{Y}) - S(\sigma_{Y})  & = 2 \gamma_0 \label{eq:ent_diff_gamma0}
\end{align} 
provided that $A, B,$ and $C$ are sufficiently large compared to the circuit depth. See Eq.~\ref{eq:fact_first}~for a self-contained derivation of Eq.~\eqref{eq:ent_diff_gamma0} starting from Definition~\ref{assumption:TEE-inv}.  (The appendices develop some lemmas related to quantum Markov chains, primarily for the purpose of deriving Eq.~\eqref{eq:ent_diff_gamma0}.)
We are finished once we prove Eq.~\eqref{eq:ent_diff_invariant}.

We now present the proof of~\eqref{eq:ent_diff_invariant}. Our main tool is the following lemma about entropy differences under reversible channels, proven in Appendix \ref{appendix:proof_ent_diff}.
\begin{lemma}\label{lemma:ent_diff-v2}
Let $\rho$ and $\rho'$ be density matrices over $PQ$ such that $\rho_Q' = \rho_Q$ and $\rho_P' = \rho_P$. Let $\mathcal{R},\mathcal{T}$ be a pair of quantum channels  $\mathcal{R} : Q \to \QQ$ and  $\mathcal{T} : \QQ \to Q$. If
\begin{align}
	\mathcal{T} \circ \mathcal{R}(\rho_{PQ}) & = \rho_{PQ} 		\label{eq:rtcond} \\
		\mathcal{T} \circ \mathcal{R}(\rho_{PQ}') & = \rho_{PQ}' \nonumber
\end{align} 
then
\begin{align}
		 S(\rho_{PQ}) - S(\rho_{PQ}')  
		= S(\mathcal{R}(\rho_{PQ})) - S(\mathcal{R}(\rho_{PQ}')).
\end{align}
\end{lemma}
\noindent
\begin{figure}[t]
    \centering
    \includegraphics[]{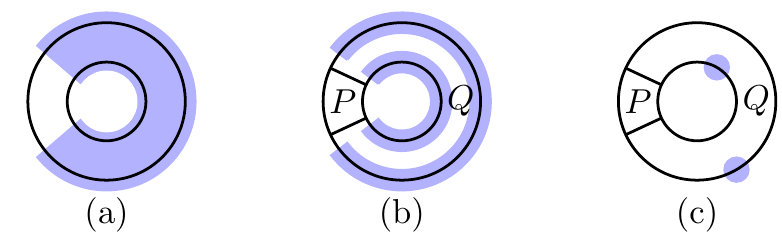}
    \caption{Using the procedure in Fig.~\ref{fig:past-causal-cone}, we remove the gates in $U$ that act deep in the interior of the annulus without changing the entropy; starting from $U$, whose support is depicted as the blue region in (a), we obtain a unitary $\overline{U}$, whose support is shown in (b). Then the support of $\overline{U}$ becomes a union of two disks, which is topologically equivalent to (c) after regrouping sites.}
    \label{fig:pp2tilde}
\end{figure}

\begin{figure}[t]
    \centering
    \begin{tikzpicture}
    \node[inner sep=5pt] (pic1) at (0,0) {\includegraphics[width=0.275\columnwidth]{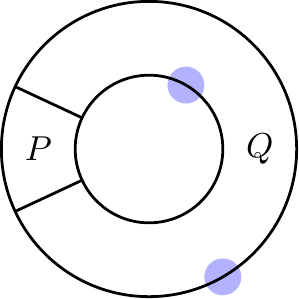}};
    \node[inner sep=5pt] (pic2) at (4,0) {\includegraphics[width=0.275\columnwidth]{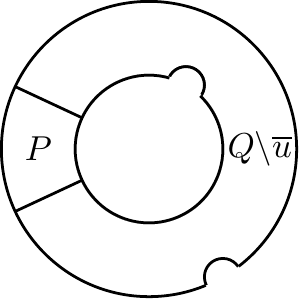}};
    \draw[->] (pic1) to node[below] {$\textrm{Tr}_{Q{\cap}\overline{u}}$} (pic2);
    \end{tikzpicture} \\
    \begin{tikzpicture}
    \node[inner sep=5pt] (pic1) at (0,0) {\includegraphics[width=0.275\columnwidth]{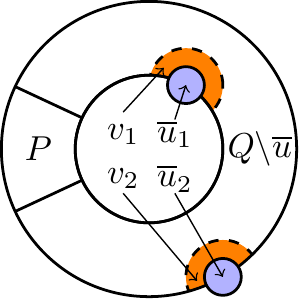}};
    \node[inner sep=5pt] (pic2) at (4,0) {\includegraphics[width=0.275\columnwidth]{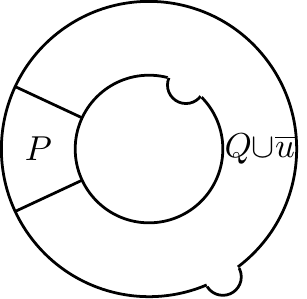}};
    \draw[->] (pic1) to node[below] {$\Phi_{v\to v\overline{u}}^{\sigma}$} (pic2);
    \end{tikzpicture}
    \caption{ 
    The construction of $\mathcal{R}$.  The first step is the partial trace $\Tr_{\overline{u}}$ over the support of $\overline{U}$, and the second step applies the Petz map $\Phi_{v \to v\overline{u}}^\sigma$ where $v=v_1 v_2$ and $\overline{u}=\overline{u}_1 \overline{u}_2$. The blue subsystem is $\overline{u}$, the support of $\overline{U}$. }
    \label{fig:iso_forward}
\end{figure}

To apply Lemma~\ref{lemma:ent_diff-v2} to our setup, we let $\rho = \overline{\lambda}$ and $\rho' = \overline{\sigma}$ where 
$\overline{\lambda} = \overline{U} \lambda \overline{U}^\dagger$ and $\overline{\sigma} = \overline{U} \sigma \overline{U}^\dagger$ and where $\overline{U}$ is a unitary obtained by removing the gates in $U$ that are deep in the interior of the annulus $ABC$ [Fig.~\ref{fig:pp2tilde}(a-b)].
We then let $P, Q$ be a partition of the annulus $ABC$ of the form shown in Fig.~\ref{fig:pp2tilde}, such that $P\subset A$ is sufficiently far away from the support of $\overline{U}$.\footnote{More precisely, $P$ must be at least two lattice spacings away from the support of $\overline{U}$; see Fig.~\ref{fig:annulus_to_chain}~for details.}
By construction, $\overline{\lambda}$ and $\overline{\sigma}$ are indistinguishable on $P$. In Appendix~\ref{appendix:max-ent_2d_new}, we show that the two states are indistinguishable on $Q$ as well [Eq.~\ref{eq:fact_second}], thus fulfilling the premise of Lemma~\ref{lemma:ent_diff-v2}.


Below we will construct quantum channels  $\mathcal{R} : Q \to \QQ$ and $\mathcal{T} : \QQ \to Q$, with $\QQ \equiv Q \cup \overline{u}$ where  $\overline{u}$ is the support of $\overline{U}$. These will obey (\ref{eq:rtcond}) with $\mathcal{R}(\overline{\lambda}_{PQ}) = \lambda_{PQ{\cup}\overline{u}}$ and $\mathcal{R}(\overline{\sigma}_{PQ}) = \sigma_{PQ{\cup}\overline{u}}$.
Because $PQ{\cup}\overline{u}=Y$ and also $S(\overline{\lambda}_{PQ}) = S(\widetilde{\lambda}_{PQ})$ and $S(\overline{\sigma}_{PQ}) = S(\widetilde{\sigma}_{PQ})$, once we construct these channels, we can immediately deduce Eq.~\eqref{eq:ent_diff_invariant} from Lemma~\ref{lemma:ent_diff-v2}. This will then complete our proof of the bound (\ref{eq:rho_sigma_CMI}), as explained earlier.

Now let us discuss our construction of $\mathcal{R}$ and $\mathcal{T}$. These maps are constructed from compositions of $\overline{U}$, partial trace, and the Petz map~\cite{Petz1987}.

To clearly render the construction of $\mathcal{R}$ and $\mathcal{T}$, we depict $\overline{u}$ as two disks of smaller sizes, as in Fig.~\ref{fig:pp2tilde}(c). Loosely speaking, $\mathcal{R}$ removes the circuits in the disks and $\mathcal{T}$ restores them.

The map $\mathcal{R}$ is constructed by applying a partial trace followed by a Petz map, best described by Fig.~\ref{fig:iso_forward}. In the first step, we trace out the region $Q \cap \overline{u}$.
This step effectively removes the circuit $\overline{U}$, mapping $\overline{\lambda}_{PQ}$ to $\lambda_{PQ{\setminus}\overline{u}}$ and $\overline{\sigma}_{PQ}$ to $\sigma_{PQ{\setminus}\overline{u}}$. In the second step, we apply the Petz map $\Phi_{v\to v\overline{u}}^{\sigma}$. We show in Appendix \ref{appendix:max-ent_2d_new}, Eq.~\ref{eq:fact_third}, that this step extends $\lambda_{PQ{\setminus}\overline{u}}$ to $\lambda_{PQ{\cup}\overline{u}}$ and $\sigma_{PQ{\setminus}\overline{u}}$ to $\sigma_{PQ{\cup}\overline{u}}$:
\begin{equation}
    \lambda_{PQ{\cup}\overline{u}} =\Phi_{v\to v\overline{u}}^{\sigma}(\lambda_{PQ{\setminus}\overline{u}}), \,\, \sigma_{PQ{\cup}\overline{u}} =\Phi_{v\to v\overline{u}}^{\sigma}(\sigma_{PQ{\setminus}\overline{u}}). \label{eq:newfact1}
\end{equation}
Combining the two steps we see that $\mathcal{R}$ maps $\overline{\lambda}_{PQ}$ to $\lambda_{PQ{\cup}\overline{u}}$ and 
$\overline{\sigma}_{PQ}$ to $\sigma_{PQ{\cup}\overline{u}}$, as required. As for the map $\mathcal{T}$, this can be constructed by simply applying $\overline{U}$ and tracing out $\overline{u} \setminus Q$. Clearly these operations map $\lambda_{PQ{\cup}\overline{u}}$ to $\overline{\lambda}_{PQ}$ and $\sigma_{PQ{\cup}\overline{u}}$ to $\overline{\sigma}_{PQ}$, as required.

\textbf{\textit{Discussion ---}}
The lower bound \eqref{eq:rho_sigma_CMI} can be generalized to the case that $\sigma$ contains an anyon in the interior of the annulus, because the proof only required the invariance of the TEE under deformations of the annulus. The lower bound then becomes $\gamma_0=\log (\mathcal{D}/d_a)$~\cite{Shi2020} for anyon $a$ in the interior. We expect similar lower bounds can be derived in a variety of setups, including systems with defects or higher dimensional systems.

Although we have only proven the lower bound \eqref{eq:gamma_nonneg} for states obtained by constant-depth circuits acting on reference states, we expect that \eqref{eq:gamma_nonneg} holds more generally. In fact, we argue heuristically in Appendix \ref{sec:stacking} that \eqref{eq:gamma_nonneg} holds for any 2D gapped bosonic ground state. The key idea is to use the fact that reference states can already realize all ``doubled’’ 2D topological phases obtained by stacking a bosonic topological phase onto its time-reversed partner \cite{Levin2006,Hahn2020,lin2021generalized}.

Our Definition~\ref{assumption:TEE-inv} and bound~\eqref{eq:rho_sigma_CMI} actually apply beyond area law states. For instance, coupling an area law reference state to a hot surface (modeled by an identity density matrix) for a short period of time cannot decrease the TEE of the joint system. We speculate that similar arguments may apply to the 3D toric code at finite temperature~\cite{Castelnovo2008}.

An interesting open question is whether our bounds apply to the TEE defined using the alternative partition in Ref.~\cite{Kitaev2006}. It would also be interesting to know whether similar results hold for a TEE defined via R\'{e}nyi entropies, which are easier to measure in quantum simulators \cite{satzinger2021realizing}. 

A final question is to understand how generically our bound (\ref{eq:gamma_nonneg}) is \emph{saturated.} How much fine-tuning is required to obtain a spurious TEE that does not decay with distance? Despite hints in this direction \cite{Zou2016, Williamson2019,stephen2019detecting,oliviero2022stability}, the general question remains open.

\begin{acknowledgments}
D.R. thanks Aram Harrow and Jonathan Sorce for discussion.
D.R. acknowledges support
from NTT (Grant AGMT DTD 9/24/20). This work was supported by the Simons Collaboration on Ultra-Quantum Matter, which is a grant from the Simons Foundation (651442, M.L.; 652264, B.S.).
T.C.L. was supported in part by funds provided by the U.S. Department of Energy (D.O.E.) under the cooperative research agreement DE-SC0009919.
\end{acknowledgments}

\bibliography{bib}

\newpage
\appendix

\section{Proof of Lemma~\ref{lemma:ent_diff-v2}}
\label{appendix:proof_ent_diff}
Here we present the proof of Lemma~\ref{lemma:ent_diff-v2}~in the main text.
\begin{proof}
	Define $\widetilde{\rho}_{P\QQ} \equiv \mathcal{R}(\rho_{PQ})$ and $\widetilde{\rho}_{P\QQ}' \equiv \mathcal{R}(\rho_{PQ}')$.
	We claim that the mutual information obeys
	\begin{equation}\label{eq:mutual-lemma-proof}
		I(P:Q)_{\rho} = I(P:\QQ)_{\widetilde{\rho}}.
	\end{equation}
   To see this, recall that mutual information is nonincreasing under quantum channels acting on one party. This monotonicity property implies two inequalities between $I(P:Q)_{\rho}$  and $I(P:\QQ)_{\widetilde{\rho}}$ in two opposite directions (using channels $\mathcal{R}$ and $\mathcal{T}$ respectively), yielding the desired equality (\ref{eq:mutual-lemma-proof}). 
   
Applying the same argument to the density matrix $\rho'$ gives
	\begin{equation}\label{eq:mutual-lemma-proof2}
		I(P:Q)_{\rho'} = I(P:\QQ)_{\widetilde{\rho}'}.
	\end{equation}
Taking the difference between (\ref{eq:mutual-lemma-proof}) and (\ref{eq:mutual-lemma-proof2}), we deduce that
   \begin{equation}
   	I(P:Q)_{\rho'}-I(P:Q)_{\rho}
   	=  I(P:\QQ)_{\widetilde{\rho}'} -  I(P:\QQ)_{\widetilde{\rho}},
  \label{eq:mutual-lemma-proof3}
   \end{equation}
To complete the argument, recall that $\rho_P= \rho'_P$ and $\rho_Q= \rho'_Q$. It follows that $\widetilde{\rho}_{\QQ}=\widetilde{\rho}'_{\QQ}$ and
   $\widetilde{\rho}_P=\widetilde{\rho}'_P$ since $\mathcal{R}$ only acts on $Q$. Substituting these equalities into (\ref{eq:mutual-lemma-proof3}), gives the desired identity
   \begin{equation}
   	S(\rho_{PQ}) - S(\rho_{PQ}')  
   	= S(\widetilde{\rho}_{P\QQ}) - S(\widetilde{\rho}_{P\QQ}').
   \end{equation}
   
\end{proof}

\section{Quantum Markov chains}
\label{sec:ssa}
In this Appendix, we collect various known facts about quantum entropies and Markov states, used extensively in this paper. We begin with the strong subadditivity of entropy (SSA), which states that~\cite{Lieb1973} 
\begin{equation}
    I(A:C|B)_{\rho}\geq 0
\end{equation}
for any tripartite density matrix $\rho_{ABC}$. A straightforward application of SSA leads to the following inequalities,
\begin{equation}
\label{eq:ssa_mono}
	\begin{aligned}
		I(AA':C|B)_{\rho} &\geq I(A':C|B)_{\rho}, \\
		I(AA':CC'|B)_{\rho} & \geq I(A:C|A'BC')_{\rho},
	\end{aligned}	
\end{equation}
for any density matrix $\rho$.

If a tripartite quantum state $\rho_{ABC}$ satisfies the SSA with equality, i.e., $I(A:C|B)_{\rho}=0$, such a state is called a quantum Markov chain~\cite{Petz1987}. An important property of quantum Markov chains is the existence of the \emph{Petz map}. If $\rho_{ABC}$ is a quantum Markov chain, there is a quantum channel $\Phi^\rho_{B\to BC}: \mathcal{B}(\mathcal{H}_B) \to \mathcal{B}(\mathcal{H}_{BC})$ such that 
\begin{equation} \label{eq:Petz_recovery_old}
    \mathcal{I}_A\otimes \Phi^\rho_{B\to BC} (\rho_{AB}) = \rho_{ABC},
\end{equation}
where $\mathcal{B}(\mathcal{H})$ is the space of bounded operators acting on $\mathcal{H}$ and $\mathcal{I}_A$ is the identity channel on $A$. We remark that it is customary to omit the superoperator $\mathcal{I}_A$ if the domain of the Petz map is obvious from context, a convention that we shall adhere to, i.e., 
\begin{equation} \label{eq:Petz_recovery}
    \Phi^\rho_{B\to BC} (\rho_{AB}) = \rho_{ABC}.
\end{equation}
The Petz map can be defined regardless of whether $\rho_{ABC}$ is a quantum Markov chain, and it depends only on the reduced density matrix $\rho_{BC}$. If $\rho_B$ is invertible (though this is not required~\cite{Hiai2011}), it admits a particularly simple form:
\begin{equation}
    \Phi_{B\to BC}^{\rho}(\cdot) = \rho_{BC}^{\frac{1}{2}} \rho_B^{-\frac{1}{2}}(\cdot) \rho_B^{-\frac{1}{2}}\rho_{BC}^{\frac{1}{2}}.
\end{equation}
Nevertheless, the precise form of the Petz map will not be important in our analysis. 

An important point that we will use below is that
\eqref{eq:Petz_recovery} holds \emph{if and only if} $I(A:C|B)_{\rho}=0$. In particular, \eqref{eq:Petz_recovery} implies that $\rho_{ABC}$ is a quantum Markov chain. This result can be derived using the fact that 
channels acting on a single party cannot increase mutual information: if (\ref{eq:Petz_recovery}) holds then $I(A:BC)_\rho \leq I(A:B)_\rho$, which implies that $I(A:C|B)_\rho \leq 0$ and hence $I(A:C|B)_\rho = 0$ by SSA.

We remark that the property \eqref{eq:Petz_recovery} of the Petz map directly implies that two locally indistinguishable quantum Markov chains are globally the same. Without loss of generality, let $\rho$ and $\sigma$ be quantum Markov chains, i.e., $I(A:C|B)_{\rho}= I(A:C|B)_{\sigma}=0$. This implies
\begin{equation}
    \begin{aligned}
        \rho_{ABC} &= \Phi_{B\to BC}^{\rho}(\rho_{AB}), \\
        \sigma_{ABC} &= \Phi_{B\to BC}^{\sigma}(\sigma_{AB}).
    \end{aligned}
\end{equation}
If $\rho_{AB} = \sigma_{AB}$ and $\rho_{BC}=\sigma_{BC}$, it follows from the property of the Petz map that $\rho_{ABC} = \sigma_{ABC}$. We therefore conclude the following.
\begin{lemma}
Let $\rho_{ABC}$ and $\sigma_{ABC}$ be states such that $I(A:C|B)_{\rho} = I(A:C|B)_{\sigma}=0$. If $\rho_{AB} = \sigma_{AB}$ and $\rho_{BC}=\sigma_{BC}$, then $\rho_{ABC} =\sigma_{ABC}$.
\label{lemma:markov_chain_local_to_global}
\end{lemma}

The notion of a quantum Markov chain can be straightforwardly generalized to $n$-partite states as well, for $n>3$. Let $\{X_1,\ldots, X_n \}$ be an \emph{ordered set} of subsystems that partition $X$. Intuitively, one can think of these subsystems as ``spins'' arranged along a one-dimensional chain, with the index representing the location of the spin. 

For the remainder of the appendices, we will frequently consider three non-empty subsystems $A, B,$ and $C$ given by contiguous subsets of this spin chain. In such cases, we say that $A, B,$ and $C$ together are \emph{chain-like} subsystems (Fig.~\ref{fig:chain_like}).

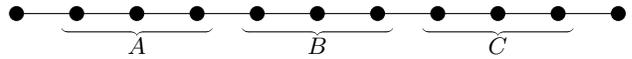
\begin{figure}[t]
  \begin{tikzpicture}[scale=0.8]
    \draw[] (0,0) -- (10,0);
    \foreach \x in {0,...,10}
             {
               \node[circle, fill=black, inner sep=2pt, minimum size=0.15cm] () at (\x, 0) {};
             }
             \draw [decorate, decoration = {calligraphic brace,mirror}] (0.75,-0.25) --  (3.25,-0.25) node[pos=0.5,black, below]{$A$};
             \draw [decorate, decoration = {calligraphic brace,mirror}] (3.75,-0.25) --  (6.25,-0.25) node[pos=0.5,black, below]{$B$};
             \draw [decorate, decoration = {calligraphic brace,mirror}] (6.75,-0.25) --  (9.25,-0.25) node[pos=0.5,black, below]{$C$};
  \end{tikzpicture}
  \caption{We call three non-empty subsystems $A$, $B$, and $C$ ``chain-like'' when they form three contiguous subsets as in the example above.
  \label{fig:chain_like}}
\end{figure}

Using the notion of chain-like subsystems, we can now present a generalization of  quantum Markov chains to $n$-partite systems.
\begin{definition}
  A quantum state $\sigma$ over ordered subsystems $\{X_1,\ldots, X_n \}$ is a Markov chain if $I(A:C|B)_{\sigma}=0$ for any chain-like subsystems $A,B,C$.
  \label{definition:markov_chain}
\end{definition}

\section{Locally Quantum Markov Chains}
\label{sec:local_vs_global_qmc}

We can also define a \emph{locally quantum Markov chain}, or locally Markov chain for short,  by only demanding the conditional mutual information constraints on proper subsystems.

\begin{definition}
  A quantum state $\sigma$ over $\{X_1,\ldots, X_n \}$ is a locally Markov chain if $I(A:C|B)_{\sigma}=0$ for any chain-like subsystems $A,B,C$ such that $ABC$ is a proper subset of $X$.
\end{definition}
\noindent
A locally quantum Markov chain need not satisfy the conditional mutual information constraint globally. That is, $I(A:C|B)\neq 0$ is allowed when $ABC = \cup_{i=1}^n X_i$.  For instance, $\rho = \frac{1}{2}(|0000\rangle\langle 0000| + |1111\rangle\langle 1111|)$ is an example of a locally Markov chain which is globally not a quantum Markov chain. In particular, a locally Markov chain is a weaker notion than quantum Markov chain.

However, this violation must be a constant for any partition of $X$ into chain-like subsystems. In fact, we show below that having a constant (possibly zero) violation is a necessary and sufficient condition for a state to be locally Markov: that is, a state $\sigma$ over $\cup_{i=1}^n X_i$ is locally Markov if and only if  $I(A:C|B)_{\sigma}$ is a constant for any chain-like $A, B,$ and $C$ that partition $X$.
\begin{proposition}
\label{proposition:lqmc_characterization}
Let $\sigma$ be a state over $X= \cup_{i=1}^n X_i$, where $n\geq 4$. Then the following statements are equivalent:
\begin{enumerate}
\item The state $\sigma$ is a locally Markov chain over $X$, i.e., $I(A:C|B)_{\sigma}=0$ for every choice of chain-like subsystems $A, B,$ and $C$ that form a proper subset of $X$.\label{item:locally_markov_1}
\item $I(A:C|B)_{\sigma}$ is a constant for any chain-like $A, B$, and $C$ that partition $X$. \label{item:locally_markov_2}
\end{enumerate}
\end{proposition}
\begin{proof}
The equivalence of the two statements hinges on a simple identity. Consider contiguous non-empty subsystems $A, B, C,$ and $D$ that partition $X$; see Fig.~\ref{fig:partition_X}. Here is the key observation:
\begin{equation}
\begin{aligned}
    I(A:C|B)_{\sigma}=0 \Leftrightarrow I(A:CD|B)_{\sigma} &= I(A:D|BC)_{\sigma}, \\
    I(B:D|C)_{\sigma}=0 \Leftrightarrow I(AB:D|C)_{\sigma} &= I(A:D|BC)_{\sigma}.
\end{aligned}
\label{eq:equivalence_key}
\end{equation}
These facts can be shown straightforwardly by expanding the conditional mutual information in terms of von Neumann entropies. As we explain below, the left and the right hand side of Eq.~\eqref{eq:equivalence_key} can be related to \ref{item:locally_markov_1} and ~\ref{item:locally_markov_2} respectively, leading to the main claim.

We first show \ref{item:locally_markov_1}$\implies$\ref{item:locally_markov_2}. From~\ref{item:locally_markov_1}, $I(A:C|B)_{\sigma}=0$ and $I(B:D|C)_{\sigma}=0$ for any non-empty and contiguous subsystems $A, B, C,$ and $D$ that partition $X$. By Eq.~\eqref{eq:equivalence_key}, $I(A:CD|B)_{\sigma} = I(A:D|BC)_{\sigma} = I(AB:D|C)_{\sigma}$.

These identities can be understood better with the help of Fig.~\ref{fig:partition_X}. Consider two different partitions of $X$, one in terms of chain-like subsystems $A, B, CD$, and the other in terms of chain-like subsystems $A, BC, D$. In both cases, $X$ is partitioned into three parts, the subsystem on the left, middle, and right, which we denote as $L$, $M$, and $R$, respectively. Viewed this way, the identity $I(A:CD|B)_{\sigma} = I(A:D|BC)_{\sigma}$ means that the conditional mutual information $I(L:R|M)$ for $X=LMR$ is invariant under shifting the boundary between the middle and the right subsystem, as long as both subsystems remain non-empty. 

Similarly, the identity $I(A:D|BC)_{\sigma} = I(AB:D|C)_{\sigma}$ means that the conditional mutual information $I(L:R|M)$ is invariant under shifting the boundary between the left and the middle subsystem, again assuming both subsystems remain non-empty. Together, we see that $I(L:R|M)_{\sigma}$ is invariant under arbitrary shifts of both boundaries, assuming the subsystems $L$, $R$, and $M$ all remain non-empty. Since any partition of $X$ into chain-like subsystems can be obtained from an appropriate shift of the two boundaries,~\ref{item:locally_markov_2} follows.

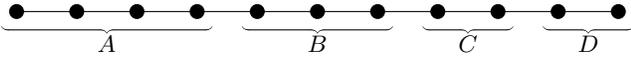
\begin{figure}[t]
  \begin{tikzpicture}[scale=0.8]
    \draw[] (0,0) -- (10,0);
    \foreach \x in {0,...,10}
             {
               \node[circle, fill=black, inner sep=2pt, minimum size=0.15cm] () at (\x, 0) {};
             }
             \draw [decorate, decoration = {calligraphic brace,mirror}] (-0.25,-0.25) --  (3.25,-0.25) node[pos=0.5,black, below]{$A$};
             \draw [decorate, decoration = {calligraphic brace,mirror}] (3.75,-0.25) --  (6.25,-0.25) node[pos=0.5,black, below]{$B$};
             \draw [decorate, decoration = {calligraphic brace,mirror}] (6.75,-0.25) --  (8.25,-0.25) node[pos=0.5,black, below]{$C$};
             \draw [decorate, decoration = {calligraphic brace,mirror}] (8.75,-0.25) --  (10.25,-0.25) node[pos=0.5,black, below]{$D$};
  \end{tikzpicture}
  \caption{Partition of $X$ into contiguous non-empty subsystems $A, B, C,$ and $D$.
  \label{fig:partition_X}}
\end{figure}

Now we show \ref{item:locally_markov_2}$\implies$\ref{item:locally_markov_1}. Without loss of generality, consider chain-like subsystems $A, B,$ and $C$, which together form a proper subset of $X$. Because $ABC$ is a proper subset, either $X_1 \not\subset A$ or $X_n\not\subset C$. 

Consider the latter case first. Without loss of generality, let $A=\cup_{k=i}^j X_k$. We can define $A' = \cup_{k=1}^j X_k\supset A$. Consider a partition of $X$ into chain-like subsystems $A', B, CD$ and also $A', BC, D$. From ~\ref{item:locally_markov_2}, we see that
\begin{equation}
    I(A':CD|B)_{\sigma} = I(A':D|BC)_{\sigma}.
\end{equation}
From (\ref{eq:equivalence_key}) we obtain
\begin{equation}
    I(A':C|B)_{\sigma}=0.
\end{equation}
Since $A\subset A'$, from Eq.~\eqref{eq:ssa_mono} we deduce that $I(A:C|B)_{\sigma}=0$. The case in which $X_1\not\subset A$ can be dealt with in a similar way, establishing  \ref{item:locally_markov_2} $\implies$ \ref{item:locally_markov_1}. 

\end{proof}

We now show that, given a locally Markov chain $\sigma$ over $n \geq 4$ subsystems,
there is a unique Markov chain state (in the sense of Definition~\ref{definition:markov_chain}) that is locally indistinguishable from $\sigma$ over every \emph{connected} proper subsystem. By connected, we mean the subsystem is a union of contiguous subsystems, e.g., $\{X_i, X_{i+1},\ldots, X_{j-1}, X_j \}$. This state, which we denote as $\tau$, is defined in the following recursive fashion
 \begin{align}
 \tau_{k+1} = \Phi_{X_k\to X_k X_{k+1}}^{\sigma}(\tau_k), \quad \forall k \in \{2,\cdots, n-1\} 
     \label{eq:tau_definition}
\end{align}
where $\tau := \tau_n$,  $\tau_2 := \sigma_{X_1X_2},$ and $\Phi_{X_k \to X_{k} X_{k+1}}^{\sigma}$ is the Petz map defined in terms of $\sigma$. 

Here are some useful properties of $\tau$, which we collate for the readers' convenience.\footnote{Some of these properties follow from~\cite[Section 5.2]{kim2021entropy}, but we provide self-contained proofs here for completeness.}
\begin{itemize}
    \item $\tau_Y = \sigma_Y$ for any connected proper subsystem $Y$. (Proposition~\ref{proposition:indistinguishability})
    \item $\tau$ is a Markov chain. (Proposition~\ref{proposition:tau_markov})
    \item Let $A, B,$ and $C$ be chain-like subsystems that partition $\{X_1,\ldots, X_n \}$. Then $\tau$ is the maximum-entropy state indistinguishable from $\sigma$ over $AB$ and $BC$. (Proposition~\ref{proposition:max_entropy})
\end{itemize}

Because $\tau$ 
is the unique maximum-entropy state indistinguishable from $\sigma$ over every connected proper subsystem, there is a sense in which $\tau$ is canonical. We thus refer to it as the \emph{canonical Markov chain} from now on. A more formal definition is stated below. 
\begin{definition}
\label{definition:canonical_markov_chain}
Let $\sigma$ be a locally Markov chain over an ordered set $\{X_1,\ldots, X_n \}$, where $n \geq 4$. The canonical Markov chain $\tau$ associated to $\sigma$ and $\{X_1,\ldots, X_n \}$ is given by
\begin{equation}
    \tau_{k+1} = \Phi_{X_k \to X_kX_{k+1}}^{\sigma}(\tau_k), \quad \forall k \in \{2,\cdots, n-1\} 
\end{equation}
where $\tau := \tau_n$ and $\tau_2 := \sigma_{X_1X_2}$. 
\end{definition}

We now prove properties of canonical Markov chains. Throughout these proofs, $\tau$ will represent the canonical Markov chain associated to a locally Markov chain $\sigma$, defined over $\{X_1,\ldots, X_n \}$.
\begin{proposition}
\label{proposition:indistinguishability}
For every connected proper subsystem $Y\subset X$, $\tau_Y = \sigma_Y$. 
\end{proposition}
Note $\tau_Y = \Tr_{X \backslash Y } \tau $ refers to the restriction of $\tau$ on $X$ to $Y \subset X$, rather than the canonical Markov chain defined with respect to $\sigma_Y$.
\begin{proof}
We first prove the claim for $Y= \cup_{i=2}^n X_i$. Because $\sigma$ is a locally Markov chain, $I(X_k:X_2\ldots X_{k-2}|X_{k-1})=0$ for every $4\leq k\leq n$. Therefore, one can obtain $\sigma_Y$ by repeatedly applying the Petz map $\Phi_{X_k\to X_{k} X_{k+1}}^{\sigma}$, from $k=3$ to $n-1$:
\begin{equation}
    \sigma_Y = \Phi_{X_{n-1}\to X_{n-1}X_n}^{\sigma}\circ\ldots \circ\Phi_{X_3\to X_3 X_4}^{\sigma}(\sigma_{X_2 X_3}).
\end{equation}

Now let us find an expression for $\tau_Y = \text{Tr}_{X_1}(\tau)$. Because all the channels appearing in the definition of $\tau$ (Definition~\ref{definition:canonical_markov_chain}) act trivially on $X_1$, the partial trace over $X_1$ can be applied before applying those channels. The resulting expression is 
\begin{equation}
    \tau_Y = \Phi_{X_{n-1}\to X_{n-1}X_n}^{\sigma}\circ\ldots \circ\Phi_{X_2\to X_2 X_3}^{\sigma}(\sigma_{X_2}).
\end{equation}
Since $\Phi_{X_2\to X_2 X_3}^{\sigma}(\sigma_{X_2}) = \sigma_{X_2 X_3}$, we derive $\tau_Y = \sigma_Y$.

Now we prove the claim for $Y'=\cup_{i=1}^{n-1} X_i$. 
Note that, because $\sigma$ is a locally Markov chain, we can relate $\tau$ and $\sigma$ in the following way:
\begin{equation}
\tau=    \Phi_{X_{n-1}\to X_{n-1}X_n}^{\sigma}(\sigma_{X_1 \dots X_{n-1}}), \label{eq:indistinguishability_helper1}
\end{equation}
Next we convert $\sigma_{X_1 \dots X_{n-1}}$ into another form, using the fact that $I(X_1:X_3\ldots X_n|X_2)_{\sigma}=0$:
\begin{equation}
    \sigma_{X_1 \dots X_{n-1}} = \Phi_{X_2\to X_1X_2}^{\sigma}(\sigma_{X_2 \dots X_{n-1}}). \label{eq:indistinguishability_helper2}
\end{equation}
Note that the channel appearing in~\eqref{eq:indistinguishability_helper2} and the channel appearing in~\eqref{eq:indistinguishability_helper1} act on disjoint subsystems, and as such, they can be exchanged. Therefore, 
\begin{equation}
\begin{aligned}
    \tau 
    &= \Phi_{X_{n-1}\to X_{n-1}X_n}^{\sigma}(\Phi_{X_2\to X_1X_2}^{\sigma}(\sigma_{X_2 \dots X_{n-1}})) \\
    &= \Phi_{X_2\to X_1X_2}^{\sigma}(\Phi_{X_{n-1}\to X_{n-1}X_n}^{\sigma}(\sigma_{X_2 \dots X_{n-1}})) \\
    &= \Phi_{X_2\to X_1X_2}^{\sigma}(\sigma_{X_2 \dots X_n}),
\end{aligned}
\end{equation}
 where in the last line we used the fact that $I(X_n:X_2\ldots X_{n-2} |X_{n-1})_{\sigma} = 0$.

Tracing out $X_n$, we obtain 
\begin{equation}
\begin{aligned}
    \tau_{Y'} &= \Phi_{X_2\to X_1X_2}^{\sigma}(\sigma_{X_2\ldots X_{n-1}}) \\
    &= \sigma_{Y'}.
\end{aligned}
\end{equation}
Since any proper connected subsystem is included in $Y$ or $Y'$, the main claim follows.

\end{proof}

\begin{proposition}
\label{proposition:tau_markov}
$\tau$ is a quantum Markov chain. 
\end{proposition}
\begin{proof}
By Proposition~\ref{proposition:indistinguishability}, $\tau$ is indistinguishable from $\sigma$ over any proper connected subsystem. Since $\sigma$ is a locally Markov chain, it follows that $I(A:C|B)_{\tau}=0$ for any chain-like subsystems $A, B,$ and $C$ forming a proper subset of $X$. Therefore, to prove the main claim, it suffices to prove $I(A:C|B)_{\tau}=0$ for chain-like subsystems $A, B,$ and $C$ that partition $X$.

Choose $C=X_n$ and $B=X_{n-1}$. Because $\sigma$ is a locally Markov chain, $\tau = \Phi_{X_{n-1}\to X_{n-1}X_n}^{\sigma}(\sigma_{X_1\ldots X_{n-1}})$. Therefore,
\begin{equation}
\begin{aligned}
    \tau_{ABC} &= \Phi_{B\to BC}^{\sigma} (\sigma_{AB}) \\
    &= \Phi_{B\to BC}^{\tau} (\tau_{AB}),
\end{aligned}
\end{equation}
where we used the fact that $\sigma_{AB} = \tau_{AB}$ and $\sigma_{BC} = \tau_{BC}$, consequences of Proposition~\ref{proposition:indistinguishability}. Therefore, $I(A:C|B)_{\tau}=0$ for this particular choice of $ABC$.   

Note that $\tau$ is in fact also a locally Markov chain, because it is indistinguishable from $\sigma$ over any proper connected subsystem and $\sigma$ is a locally Markov chain. By Proposition~\ref{proposition:lqmc_characterization}, $I(A:C|B)_{\tau}$ is a constant for any chain-like $A, B,$ and $C$ that partition $X$. Moreover, it was shown above that this constant is $0$. Therefore, $I(A:C|B)_{\tau}=0$ for any chain-like subsystems $A, B,$ and $C$. This completes the proof. 
\end{proof}

\begin{proposition}
\label{proposition:max_entropy}
Let $A, B,$ and $C$ be chain-like subsystems that partition $\{ X_1,\ldots, X_n\}$. Then $\tau$ is the maximum-entropy state indistinguishable from $\sigma$ over $AB$ and $BC$.
\end{proposition}
\begin{proof}
Suppose there exists another such state, say $\tau'$. Partition $X$ into $ABC$ which are chain-like. $I(A:C|B)_{\tau}=0$. By the maximum-entropy condition, $S(\tau') \geq S(\tau)$. Because $\tau$ and $\tau'$ are indistinguishable over $AB$, $BC$, and $B$, we conclude $I(A:C|B)_{\tau'}\leq 0$. However, by SSA $I(A:C|B)_{\tau'}\geq 0$. Hence, $I(A:C|B)_{\tau'} = 0$. 

Therefore, we obtain the following expression for $\tau$ and $\tau'$:
\begin{equation}
    \begin{aligned}
        \tau &= \Phi_{B\to BC}^{\tau}(\tau_{AB}) \\
        \tau' &= \Phi_{B\to BC}^{\tau'}(\tau_{AB}').
    \end{aligned}
\end{equation}
Because $\tau_{AB} = \tau_{AB}'$ and $\tau_{BC} = \tau_{BC}'$, we deduce $\tau = \tau'$. This completes the proof. 
\end{proof}

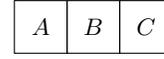
\begin{figure}[t]
    \centering
    \begin{tikzpicture}[scale=0.7]
        \draw[] (0,0) -- ++ (1,0) -- ++ (0, 1) -- ++ (-1, 0) -- cycle;
        \draw[] (1,0) -- ++ (1,0) -- ++ (0, 1) -- ++ (-1, 0) -- cycle;
        \draw[] (2,0) -- ++ (1,0) -- ++ (0, 1) -- ++ (-1, 0) -- cycle;
        \node[] () at (0.5, 0.5) {$A$};
        \node[] () at (1.5, 0.5) {$B$};
        \node[] () at (2.5, 0.5) {$C$};
    \end{tikzpicture}
    \caption{A chain-like region $ABC$ in 2D.}
    \label{fig:chain_like_ex}
\end{figure}

\section{Chain-like subsystems in 2D}
\label{sec:chain_like_2D}

In this Appendix, we prove a certain fact about \emph{chain-like} subsystems in 2D. Roughly speaking, if three subsystems $A$, $B$, and $C$ are arranged in a form shown in Fig.~\ref{fig:chain_like_ex}, we say these subsystems are chain-like. Our main claim is that for such subsystems $I(A:C|B)_{\sigma}=0$. For concreteness, we provide a precise definition and a statement below.
\begin{definition}[Chain-like regions in 2D] Subsystems $A$, $B$, $C$ of a 2D disk (or plane) form a chain-like region if their relation is topologically equivalent to the regions $A$, $B$ and $C$ shown in Fig.~\ref{fig:chain_like_ex}.  
\end{definition}

\begin{figure}[t]
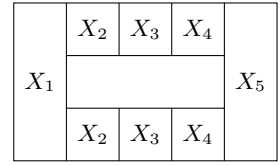

\include{fig_annulus_to_chain}
 	\caption{We can view an annulus in 2D as a 1D ``spin chain.'' The shown partition defines a spin chain over the ordered subsystems $\{ X_1, X_2, X_3, X_4, X_5 \}$. }
 	\label{fig:annulus_to_chain}
 \end{figure}

\begin{proposition}[Chain-like regions in 2D are Markov] \label{cor:reference_is_Markov}
Let $\sigma$ be a reference state as in Definition~\ref{assumption:TEE-inv} of the main text. For any chain-like region ABC in 2D (as in Fig.~\ref{fig:chain_like_ex}), 
we have $I(A:C|B)_{\sigma}=0$.
\end{proposition}

\begin{proof}
To prove this claim, we shall identify the chain-like region $ABC$ as a subset of an annulus. Specifically, given any such $ABC$, consider an auxiliary annulus as shown in Fig.~\ref{fig:annulus_to_chain} such that $A, B,$ and $C$ are precisely the upper connected components of $X_2,$ $X_3,$ and $X_4,$ respectively.

Note that, by the definition of the reference state (Definition~\ref{assumption:TEE-inv}~of the main text), the state $\sigma$ over $\{X_1, X_2, X_3, X_4, X_5 \}$ has a constant value of $I(P:Q|R)_{\sigma}$ for any contiguous $P, R$, and $Q$
that partition $X=\cup_{i=1}^5 X_i$, i.e.,\ $I(X_1 X_2 : X_5 | X_3 X_4)=I(X_1  : X_5 | X_2 X_3 X_4)=\ldots$ etc. By Proposition~\ref{proposition:lqmc_characterization}, we thus conclude that $\sigma$ is a locally Markov chain over $\{X_1, X_2, X_3, X_4, X_5 \}$. In particular,
\begin{equation}
    I(X_2:X_4|X_3)_{\sigma}=0.
\end{equation}

Because the reference state $\sigma$ has vanishing mutual information between any two subsystems that are not adjacent to each other, the state over the upper and the lower components (of any region in $\{ X_2X_3, X_3 X_4, X_3, X_2 X_3 X_4 \}$) factorize. Therefore, 
\begin{equation}
    I(A:C|B)_{\sigma}=0,
\end{equation}
where again we identify $A,B,C$ with the upper connected components of $X_2$, $X_3$, and $X_4$, respectively.
\end{proof}

\section{Max-entropy state on a 2D annulus}
\label{appendix:max-ent_2d_new}

The argument presented in the main text was based on three facts. Specifically, we considered a partition of an annulus $ABC$ into $P$ and $Q=ABC{\setminus}P$ such that $P\subset A$ is sufficiently far from the support of the unitary $U$ (Fig.~\ref{fig:partition_annulus_comparison}).\footnote{More specifically, we assumed that $P$ must be at least two lattice spacing away from the support of $\overline{U}$ (Fig.~\ref{fig:partition_regroup}). Because the support of $\overline{U}$ is included in the support of $U$, the requisite condition can be fulfilled by demanding $P$ is at least two lattice spacing away from the support of $U$.}
The density matrix $\lambda$ was then defined as the maximum-entropy state consistent with $\sigma$ over the past light cones of $AB$ and $BC$ (Fig.~\ref{fig:past_light_cones_abc}). Denoting the support of $U$ as $u$, we claimed that 
\begin{align}
    S(\lambda_{PQ{\cup}u}) - S(&\sigma_{PQ{\cup}u}) = 2\gamma_0, \label{eq:fact_first}\\
    \overline{\lambda}_P = \overline{\sigma}_P, \,&\,\overline{\lambda}_Q = \overline{\sigma}_Q \label{eq:fact_second}.
\end{align}
Moreover, we also claimed that
\begin{equation}
\begin{aligned}
    \lambda_{PQ{\cup}\overline{u}} &=\Phi_{v\to v\overline{u}}^{\sigma}(\lambda_{PQ{\setminus}\overline{u}}), \\
    \sigma_{PQ{\cup}\overline{u}} &=\Phi_{v\to v\overline{u}}^{\sigma}(\sigma_{PQ{\setminus}\overline{u}}), 
\end{aligned}\label{eq:fact_third}
\end{equation}
where $\overline{U}$ is the unitary obtained by removing the gates in $U$ that are deep inside the annulus, $\overline{u}$ is its support, and $v$ is the region in the annulus that surrounds $\overline{u}$ (Fig.~\ref{fig:iso_forward}~of the main text). Below, we prove these facts.

The first two facts Eq.~\eqref{eq:fact_first} and~\eqref{eq:fact_second} follow straightforwardly from the observation that the annulus $PQ \cup u$ can be partitioned into a set of subsystems in such a way that the reference state $\sigma$ is a locally Markov chain over these subsystems (Appendix~\ref{sec:local_vs_global_qmc}). This implies that $\lambda$ is the canonical Markov chain (Definition~\ref{definition:canonical_markov_chain}) associated with $\sigma$ over this partition. Then from the properties of canonical Markov chains, Eq.~\eqref{eq:fact_first} and~\eqref{eq:fact_second} both follow.

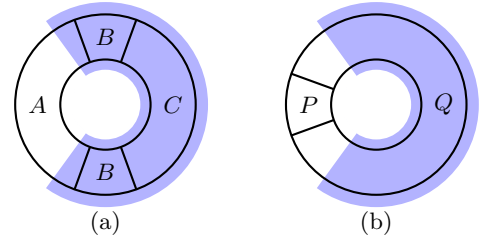
\begin{figure}[t]
    \centering
    \begin{tikzpicture}[every path/.style={thick}, scale=0.6]
        \filldraw[blue!30!white] (0,0) circle (2.25cm);
	\filldraw[white] (0,0) circle (0.75cm);
		\filldraw[white] (126:0.75cm) -- (126:2.5cm) -- (126:2.5cm) arc (126:234:2.5cm) -- (234:0.75cm) -- (234:0.75cm) arc (234:150:0.75cm) --cycle; 
    \draw[] (0,0) circle (1cm);
    \draw[] (0,0) circle (2cm);

    \draw[] (70:1cm) -- (70:2cm);
    \draw[] (-70:1cm) -- (-70:2cm);
    \draw[] (110:1cm) -- (110:2cm);
    \draw[] (-110:1cm) -- (-110:2cm);
    
    \node[] () at (0, -2.625cm) {(a)};
    \node[] () at (180:1.5cm) {$A$};
    \node[] () at (0:1.5cm) {$C$};
    \node[] () at (90:1.5cm) {$B$};
    \node[] () at (-90:1.5cm) {$B$};
    \begin{scope}[xshift=6cm]
        \filldraw[blue!30!white] (0,0) circle (2.25cm);
	\filldraw[white] (0,0) circle (0.75cm);
		\filldraw[white] (126:0.75cm) -- (126:2.5cm) -- (126:2.5cm) arc (126:234:2.5cm) -- (234:0.75cm) -- (234:0.75cm) arc (234:150:0.75cm) --cycle; 
    \draw[] (0,0) circle (1cm);
    \draw[] (0,0) circle (2cm);

    \draw[] (160:1cm) -- (160:2cm);
    \draw[] (-160:1cm) -- (-160:2cm);
    
    \node[] () at (0, -2.625cm) {(b)};
    \node[] () at (180:1.5cm) {$P$};
    \node[] () at (0:1.5cm) {$Q$};
    \end{scope}
    \end{tikzpicture}
    \caption{Two different partitions of the annulus, one into $ABC$ (a) and one into $PQ$ (b). $P\subset A$ is chosen so that it is sufficiently far from the support of the unitary $U$ (blue).}
    \label{fig:partition_annulus_comparison}
\end{figure}

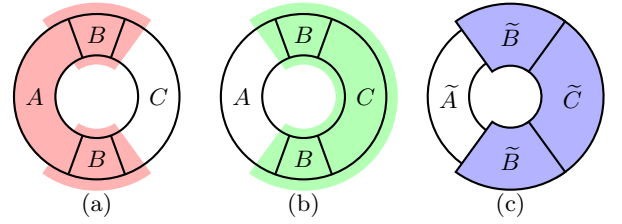
\begin{figure}[t]
    \centering
    \begin{tikzpicture}[every path/.style={thick}, scale=0.55]
        \filldraw[red!30!white] (0,0) circle (2.25cm);
	\filldraw[white] (0,0) circle (0.75cm);
		\filldraw[white] (126:0.75cm) -- (126:2.5cm) -- (126:2.5cm) arc (126:234:2.5cm) -- (234:0.75cm) -- (234:0.75cm) arc (234:150:0.75cm) --cycle; 
  \filldraw[white] (54:0.75cm) -- (54:2.5cm) -- (54:2.5cm) arc (54:-54:2.5cm) -- (-54:0.75cm) -- (-54:0.75cm) arc (-54:150:0.75cm) --cycle; 

  \filldraw[red!30!white] (120:1cm) -- (120:2cm) -- (120:2cm) arc (120:240:2cm) -- (240:1cm) -- (240:1cm) arc (240:120:1cm) -- cycle;
    \draw[] (0,0) circle (1cm);
    \draw[] (0,0) circle (2cm);

    \draw[] (70:1cm) -- (70:2cm);
    \draw[] (-70:1cm) -- (-70:2cm);
    \draw[] (110:1cm) -- (110:2cm);
    \draw[] (-110:1cm) -- (-110:2cm);
    
    \node[] () at (0, -2.625cm) {(a)};
    \node[] () at (180:1.5cm) {$A$};
    \node[] () at (0:1.5cm) {$C$};
    \node[] () at (90:1.5cm) {$B$};
    \node[] () at (-90:1.5cm) {$B$};
    \begin{scope}[xshift=5cm]
        \filldraw[green!30!white] (0,0) circle (2.25cm);
	\filldraw[white] (0,0) circle (0.75cm);
		\filldraw[white] (126:0.75cm) -- (126:2.5cm) -- (126:2.5cm) arc (126:234:2.5cm) -- (234:0.75cm) -- (234:0.75cm) arc (234:150:0.75cm) --cycle; 
    \draw[] (0,0) circle (1cm);
    \draw[] (0,0) circle (2cm);

    \draw[] (70:1cm) -- (70:2cm);
    \draw[] (-70:1cm) -- (-70:2cm);
    \draw[] (110:1cm) -- (110:2cm);
    \draw[] (-110:1cm) -- (-110:2cm);
    
    \node[] () at (0, -2.625cm) {(b)};
    \node[] () at (180:1.5cm) {$A$};
    \node[] () at (0:1.5cm) {$C$};
    \node[] () at (90:1.5cm) {$B$};
    \node[] () at (-90:1.5cm) {$B$};
    \end{scope}
    \begin{scope}[xshift=10cm]
    \draw[] (120:1cm) -- (120:2cm) -- (120:2cm) arc (120:240:2cm) -- (240:1cm) -- (240:1cm) arc (240:120:1cm) -- cycle;
    \draw[fill=blue!30!white] (60:0.75cm) -- (60:2.25cm) -- (60:2.25cm) arc (60:-60:2.25cm) -- (-60:0.75cm) -- (-60:0.75cm) arc (-60:60:0.75cm) -- cycle;
    \draw[fill=blue!30!white] (54:0.75cm) -- (54:2.25cm) -- (54:2.25cm) arc (54:126:2.25cm) -- (126:0.75cm) -- (126:0.75cm) arc (126:54:0.75cm) -- cycle;
    \draw[fill=blue!30!white] (-54:0.75cm) -- (-54:2.25cm) -- (-54:2.25cm) arc (-54:-126:2.25cm) -- (-126:0.75cm) -- (-126:0.75cm) arc (-126:-54:0.75cm) -- cycle;

    \node[] () at (0, -2.625cm) {(c)};
    \node[] () at (180:1.5cm) {$\widetilde{A}$};
    \node[] () at (0:1.5cm) {$\widetilde{C}$};
    \node[] () at (90:1.5cm) {$\widetilde{B}$};
    \node[] () at (-90:1.5cm) {$\widetilde{B}$};
    \end{scope}
    \end{tikzpicture}
    \caption{(a) Past light cone of $AB$ (red). (b) Past light cone of $BC$ (green). (c) The enlarged annulus $PQ{\cup}u$ can be partitioned into the new subsystems $\widetilde{A}, \widetilde{B},$ and $\widetilde{C}$. The past light cone of $AB$ and $BC$ is $\widetilde{A}\widetilde{B}$ and $\widetilde{B}\widetilde{C}$, respectively. The blue region is the support of $U$.}
    \label{fig:past_light_cones_abc}
\end{figure}

More specifically, consider the partition shown in Fig.~\ref{fig:past_light_cones_abc}(c). Here $\widetilde{B}$ is the past light cone of $B$. The remaining subsystems $\widetilde{A}$ and $\widetilde{C}$ are obtained by subtracting $\widetilde{B}$ from the past light cone of $AB$ and $BC$, respectively. Together, $\widetilde{A}$, $\widetilde{B}$, and $\widetilde{C}$ partition the enlarged annulus $ABC{\cup}u = PQ{\cup}u$. 

By our assumption, $P$ is separated from the support of $U$. Since the support of $U$ includes  $\widetilde{B}\widetilde{C}$, $\widetilde{B}\widetilde{C}$ is separated from $P$. Therefore, 
the enlarged annulus can be further partitioned into $P, \widetilde{A}{\setminus}P, \widetilde{B}, \widetilde{C}$ (Fig.~\ref{fig:partition_enlarged}). With respect to this partition, we now invoke facts about locally Markov chains (Appendix~\ref{sec:local_vs_global_qmc}). To draw a parallel with the discussion in Appendix~\ref{sec:local_vs_global_qmc}, let us momentarily denote $P, \widetilde{A}{\setminus}P, \widetilde{B}, \widetilde{C}$ as $X_1, X_2, X_3, X_4$, respectively. Because we assumed that the reference state $\sigma$ has a constant TEE, it immediately follows that $\sigma_{X_1X_2X_3X_4}$ is a locally Markov chain with respect to the partition $\{X_1, X_2, X_3, X_4 \}$.

Since $\sigma$ is a locally Markov chain, there is also a canonical Markov chain associated to $\sigma$. Crucially, this canonical Markov chain is precisely $\lambda$. Recall that the state $\lambda$ was defined as the maximum-entropy state consistent with the past light cone of $AB$ and $BC$. Thus, by Proposition~\ref{proposition:max_entropy}, we immediately see that $\lambda$ is the canonical Markov chain associated with $\sigma$, with respect to the partition $\{X_1, X_2, X_3, X_4 \}$. 

Now that we established that $\lambda$ is the canonical Markov chain, we can use our machinery to prove Eqs.~\eqref{eq:fact_first} and~\eqref{eq:fact_second}. In particular, because the canonical Markov chain is a Markov chain (Proposition~\ref{proposition:tau_markov}), we know $I(X_1 X_2 : X_4 | X_3)_\lambda = 0$. Then \eqref{eq:fact_first} follows from the fact that $I(X_1 X_2 : X_4 | X_3)_\sigma = 2 \gamma_0$, and that $\lambda_{X_1 X_2 X_3} = \sigma_{X_1 X_2 X_3}$ and $\lambda_{X_3 X_4} = \sigma_{X_3 X_4}$ (Proposition~\ref{proposition:indistinguishability}). Similarly, to establish \eqref{eq:fact_second}, we use Proposition~\ref{proposition:indistinguishability} to deduce $\lambda_{X_1} = \sigma_{X_1}$ and $\lambda_{X_2X_3X_4} = \sigma_{X_2X_3X_4}$. Then since $X_1$ is the past light cone of $P$ and $X_2 X_3 X_4$ includes the past light cone of $Q$, we derive \eqref{eq:fact_second}.

\begin{figure}[t]
    \centering
    \begin{tikzpicture}[every path/.style={thick}, scale=0.6]
    \draw[] (120:1cm) -- (120:2cm) -- (120:2cm) arc (120:240:2cm) -- (240:1cm) -- (240:1cm) arc (240:120:1cm) -- cycle;
    \draw[fill=blue!30!white] (60:0.75cm) -- (60:2.25cm) -- (60:2.25cm) arc (60:-60:2.25cm) -- (-60:0.75cm) -- (-60:0.75cm) arc (-60:60:0.75cm) -- cycle;
    \draw[fill=blue!30!white] (54:0.75cm) -- (54:2.25cm) -- (54:2.25cm) arc (54:126:2.25cm) -- (126:0.75cm) -- (126:0.75cm) arc (126:54:0.75cm) -- cycle;
    \draw[fill=blue!30!white] (-54:0.75cm) -- (-54:2.25cm) -- (-54:2.25cm) arc (-54:-126:2.25cm) -- (-126:0.75cm) -- (-126:0.75cm) arc (-126:-54:0.75cm) -- cycle;
  
    \draw[] (160:1cm) -- (160:2cm);
    \draw[] (200:1cm) -- (200:2cm);
    
    \node[] () at (180:1.5cm) {$P$};
    \node[] () at (0:1.5cm) {$\widetilde{C}$};
    \node[] () at (90:1.5cm) {$\widetilde{B}$};
    \node[] () at (-90:1.5cm) {$\widetilde{B}$};
    \node[] (sub) at (-4.25cm,0cm) {$\widetilde{A}{\setminus}P$};

    \draw[->] (sub) -- (-4.25, 1.25) -- (-1.75 , 1.25);
    \draw[->] (sub) -- (-4.25, -1.25) -- (-1.75 , -1.25);
    \end{tikzpicture}
    \caption{Partition of an enlarged annulus $PQ{\cup}u$ into $P, \widetilde{A}{\setminus}P, \widetilde{B},$ and $\widetilde{C}$.}
    \label{fig:partition_enlarged}
\end{figure}
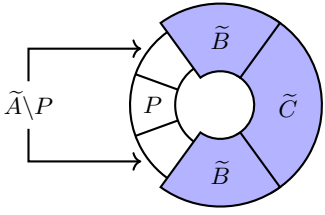

The proof of~\eqref{eq:fact_third} involves two different partitions of the enlarged annulus $PQ{\cup}u$. The first partition is described in Fig.~\ref{fig:partition_fine_grained}, which is obtained by dividing $\widetilde{A}{\setminus}P$ further into $Q_1$ and $\widetilde{A}{\setminus}(PQ_1)$. The second partition is based on the subsystems defined in terms of $\overline{U}$, which is the unitary obtained from $U$ by removing the gates deep in the annulus. Provided that the annulus is sufficiently thick, the support of $\overline{U}$ can be split into two non-overlapping disks, each of which can be furthermore surrounded by non-overlapping disks lying inside the annulus [Fig.~\ref{fig:partition_regroup}(a-b)]. 

\begin{figure}[t]
    \centering
    \begin{tikzpicture}[every path/.style={thick}, scale=0.6]
    \draw[] (120:1cm) -- (120:2cm) -- (120:2cm) arc (120:240:2cm) -- (240:1cm) -- (240:1cm) arc (240:120:1cm) -- cycle;
    \draw[fill=blue!30!white] (60:0.75cm) -- (60:2.25cm) -- (60:2.25cm) arc (60:-60:2.25cm) -- (-60:0.75cm) -- (-60:0.75cm) arc (-60:60:0.75cm) -- cycle;
    \draw[fill=blue!30!white] (54:0.75cm) -- (54:2.25cm) -- (54:2.25cm) arc (54:126:2.25cm) -- (126:0.75cm) -- (126:0.75cm) arc (126:54:0.75cm) -- cycle;
    \draw[fill=blue!30!white] (-54:0.75cm) -- (-54:2.25cm) -- (-54:2.25cm) arc (-54:-126:2.25cm) -- (-126:0.75cm) -- (-126:0.75cm) arc (-126:-54:0.75cm) -- cycle;
  
    \draw[] (160:1cm) -- (160:2cm);
    \draw[] (200:1cm) -- (200:2cm);

    \draw[] (142.5:1cm) -- (142.5:2cm);
    \draw[] (217.5:1cm) -- (217.5:2cm);
    
    \node[] () at (180:1.5cm) {$P$};
    \node[] () at (0:1.5cm) {$\widetilde{C}$};
    \node[] () at (90:1.5cm) {$\widetilde{B}$};
    \node[] () at (-90:1.5cm) {$\widetilde{B}$};
    \node[] (sub) at (-5cm,0) {$\widetilde{A}{\setminus}(PQ_1)$};
    \node[] (q1) at (-3cm, 0) {$Q_1$};

    \draw[->] (q1) -- (-3, 0.875) -- (-1.875, 0.875);
    \draw[->] (q1) -- (-3, -0.875) -- (-1.875, -0.875);

    \draw[->] (sub) -- (-5, 1.5) -- (-1.5, 1.5);
    \draw[->] (sub) -- (-5, -1.5) -- (-1.5, -1.5);
    \end{tikzpicture}
    \caption{Partition of an enlarged annulus $PQ{\cup}u$ into $P, Q_1, \widetilde{A}{\setminus}(PQ_1), \widetilde{B},$ and $\widetilde{C}$.}
    \label{fig:partition_fine_grained}
\end{figure}
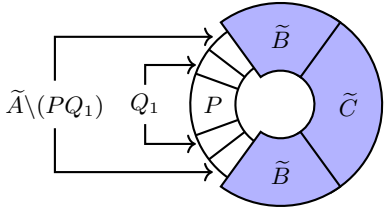

\begin{figure}[t]
    \centering
    \begin{tikzpicture}[every path/.style={thick}, scale=0.75]

    \filldraw[orange, dashed, draw=black] (130:1.625cm) -- (130:2cm) -- (130:2cm) arc (130:-130:2cm) -- (-130:1.625cm) -- (-130:1.625cm) arc (-130:130:1.625cm) -- cycle;
    \filldraw[blue!30!white] (126:1.75cm) -- (126:2.25cm) -- (126:2.25cm) arc (126:-126:2.25cm) -- (-126:1.75cm) -- (-126:1.75cm) arc (-126:126:1.75cm) -- cycle;

    \filldraw[orange, dashed, draw=black] (130:1cm) -- (130:1.375cm) -- (130:1.375cm) arc (130:-130:1.375cm) -- (-130:1cm) -- (-130:1cm) arc (-130:130:1cm) -- cycle;
    \filldraw[blue!30!white] (126:1cm) -- (126:1.25cm) -- (126:1.25cm) arc (126:-126:1.25cm) -- (-126:1cm) -- (-126:1cm) arc (-126:126:1cm) -- cycle;
    

    \draw [] (0,0) circle [radius=1];
    \draw [] (0,0) circle [radius=2];
    \draw[] (160:1cm) -- (160:2cm);
    \draw[] (200:1cm) -- (200:2cm);

    \draw[] (142.5:1cm) -- (142.5:2cm);
    \draw[] (217.5:1cm) -- (217.5:2cm);
    
    \node[] () at (180:1.5cm) {$P$};
    \node[] (q1) at (-3cm, 0) {$Q_1$};

    \draw[->] (q1) -- (-3, 0.875) -- (-1.875, 0.875);
    \draw[->] (q1) -- (-3, -0.875) -- (-1.875, -0.875);

    \begin{scope}[xshift=5.5cm]
    \draw[dashed, thick, fill=orange] (60:1) circle (0.5);
 \begin{scope}
 \tikzset{every path/.style={}}
    \clip (0,0) circle (2);
   \draw[dashed, thick, fill=orange] (300:2) circle (0.5);
 \end{scope}
 \draw[thick, fill=white] (0,0) circle (1);
 \draw[thick] (0,0) circle (1);
 \draw[thick] (0,0) circle (2);
 \draw[thick, fill=blue!30] (300:2) circle (0.25);
 \draw[thick, fill=blue!30] (60:1) circle (0.25);
 \draw[thick] (-155:1) -- (-155:2);
 \draw[thick] (155:1) -- (155:2);
 \node[] (P) at (180:1.5) {$P$};
 \node[] () at (135:1.5) {$Q_1$};
 \node[] () at (-135:1.5) {$Q_1$};

 \draw[] (120:1) -- (120:2);
 \draw[] (-120:1) -- (-120:2);
 \node[] (Q) at (0:1.5) {$Q_2{\setminus}\overline{u}$};
 \node[] (B1) at (-0.35, 0.2) {$v_1$};
 \node[] (C1) at (0.35, 0.2) {$\overline{u}_1$};
 \node[] (B2) at (-0.35, -0.4) {$v_2$};
 \node[] (C2) at (0.35, -0.4) {$\overline{u}_2$};
 \draw[->] (-0.35, 0.5) -- (0.2,1.1);
 \draw[->] (-0.35, -0.6) -- (290:1.9);
 \draw[->] (0.35, 0.4) -- (60: 1);
 \draw[->] (0.35, -0.6) -- (300:2);
    \end{scope}
    \end{tikzpicture}
    \caption{(a) If the annulus is sufficiently thick, upon removing the gates deep inside the annulus, we obtain a unitary $\overline{U}$ acting on two non-overlapping disks (blue). Furthermore, one can choose two disconnected disks in the annulus that surround each disk (orange). Figure (b) can be obtained from (a) by an appropriate regrouping of sites and topologically equivalent re-drawing. Here $Q=Q_1Q_2$, $\overline{u} = \overline{u}_1 \overline{u}_2$, and $v=v_1v_2$. To be able to partition the annulus this way, the support of $\overline{U}$ must be at least two lattice spacing away from $P$.}
    \label{fig:partition_regroup}
\end{figure}
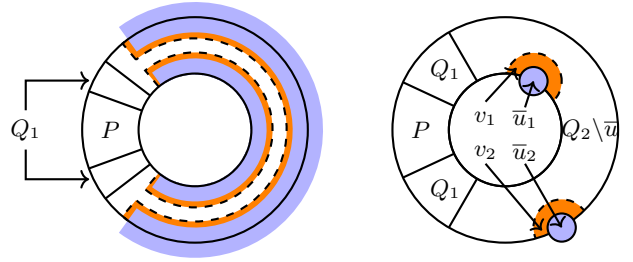

We use the first partition to establish the following fact:
\begin{equation}
\lambda_{PQ{\cup}\overline{u}} = \Phi_{Q_1\to PQ_1}^{\sigma}(\sigma_{Q{\cup}\overline{u}}). \label{eq:lambda_recovery}
\end{equation}
We can view $\lambda$ as the canonical Markov chain associated with the reference state $\sigma$ with respect to the partition $\{X_1, X_2, X_3, X_4, X_5\} = \{P, Q_1, \widetilde{A}{\setminus}(PQ_1), \widetilde{B}, \widetilde{C} \}$. This implies that $\lambda$ is a Markov chain satisfying $I(X_1:X_3X_4X_5|X_2)_{\lambda}=0$ (Proposition~\ref{proposition:tau_markov}). Furthermore, $\lambda$ and $\sigma$ are indistinguishable over $X_1X_2 = PQ_1$ and $X_2X_3X_4X_5=\widetilde{A}\widetilde{B}\widetilde{C}{\setminus}P = Q{\cup}\overline{u}$ (Proposition~\ref{proposition:indistinguishability}). Eq.~\eqref{eq:lambda_recovery} follows directly from these facts.

We use the second partition to prove the following fact:
\begin{equation}
    I(\overline{u}:Q{\setminus} (v\overline{u})|v)_{\sigma}=0.\label{eq:cmi_fact_key}
\end{equation}
To see why this is true, note that 
\begin{equation}
    \begin{aligned}
    I(\overline{u}_1:Q{\setminus}(v_1\overline{u})|v_1)_{\sigma}&=0, \\
        I(\overline{u}_2:(Q{\cup}\overline{u}_1){\setminus}(v_2\overline{u}_2) |v_2)_{\sigma}&=0,
    \end{aligned}
\end{equation}
because $\overline{u}_1, v_1,$ and $Q{\setminus}(v_1\overline{u})$ is chain-like and so is $\overline{u}_2, v_2,$ and $(Q{\cup}\overline{u}_1){\setminus}(v_2\overline{u}_2)$ (Proposition~\ref{cor:reference_is_Markov}). More explicitly, these equations can be written as
\begin{equation}
\begin{aligned}
&I(\overline{u}_1:Q{\setminus}(v_1\overline{u})|v_1)_{\sigma} = \\
&S(\sigma_{\overline{u}_1v_1}) + S(\sigma_{Q{\setminus}\overline{u}}) - S(\sigma_{v_1}) - S(\sigma_{(Q{\cup}\overline{u}_1){\setminus}\overline{u}_2}) \\
&=0
\end{aligned}
\end{equation}
and 
\begin{equation}
\begin{aligned}
&I(\overline{u}_2:(Q{\cup}\overline{u}_1){\setminus}(v_2\overline{u}_2) |v_2)_{\sigma} = \\
& S(\sigma_{\overline{u}_2v_2}) + S(\sigma_{(Q{\cup}\overline{u}_1){\setminus}\overline{u}_2}) - S(\sigma_{v_2}) - S(\sigma_{(Q{\cup}\overline{u})})\\
&=0.
\end{aligned}
\end{equation}
Summing the two, we obtain
\begin{equation}
I(\overline{u}:Q{\setminus}(v\overline{u})|v)_{\sigma} + I(v_1\overline{u}_1:v_2\overline{u}_2) - I(v_1:v_2)=0.
\end{equation}
Because the mutual information between two non-adjacent subsystems is zero, we conclude $I(\overline{u}:Q{\setminus}(v\overline{u})|v)_{\sigma}=0$. 

Combining these two facts, we now prove~\eqref{eq:fact_third}, starting with the first identity. First note that, upon taking a partial trace on $\overline{u}$, we get
\begin{equation}
\begin{aligned}
\textrm{Tr}_{\overline{u}}(\lambda_{PQ{\cup}\overline{u}}) &= \textrm{Tr}_{\overline{u}}(\Phi_{Q_1\to PQ_1}^{\sigma}(\sigma_{Q{\cup}\overline{u}})) \\
&= \Phi_{Q_1\to PQ_1}^{\sigma}(\textrm{Tr}_{\overline{u}}(\sigma_{Q{\cup}\overline{u}})) \\
&= \Phi_{Q_1\to PQ_1}^{\sigma}(\sigma_{Q{\setminus}\overline{u}}).
\end{aligned}
\end{equation} 
On the second line, the map $\text{Tr}_{\overline{u}}$  and $\Phi_{Q_1\to PQ_1}^{\sigma}$ can be exchanged because they act on disjoint subsystems. Therefore, we obtain
\begin{equation}
\lambda_{PQ{\setminus}\overline{u}} = \Phi_{Q_1\to PQ_1}^{\sigma}(\sigma_{Q{\setminus}\overline{u}}).
\end{equation}
Next note that, by construction, the Petz maps $\Phi_{v\to v\overline{u}}^{\sigma}$ and $\Phi_{Q_1\to PQ_1}^{\sigma}$ act on disjoint subsystems. As such, their order can be exchanged. We therefore obtain
\begin{equation}
\begin{aligned}
\Phi_{v\to v\overline{u}}^{\sigma} (\lambda_{PQ{\setminus}{\overline{u}}}) &= \Phi_{v\to v\overline{u}}^{\sigma}(\Phi_{Q_1\to PQ_1}^{\sigma}(\sigma_{Q\setminus{\overline{u}}})) \\
&=\Phi_{Q_1\to PQ_1}^{\sigma}(\Phi_{v\to v\overline{u}}^{\sigma}(\sigma_{Q\setminus{\overline{u}}})) \\
&= \Phi_{Q_1\to PQ_1}^{\sigma}(\sigma_{Q{\cup}\overline{u}}) \\
&= \lambda_{PQ{\cup}\overline{u}}.
\end{aligned}
\end{equation}
In the second line we exchanged the order of the channels. In the third and the fourth line, we used ~\eqref{eq:cmi_fact_key} and~\eqref{eq:lambda_recovery}, respectively. This establishes the first identity in~\eqref{eq:fact_third}.

Now we prove the second identity in~\eqref{eq:fact_third}. The key point is that
\begin{equation}
I(\overline{u}:PQ{\setminus}(v\overline{u})|v)_{\sigma}=0,\label{eq:cmi_reference}
\end{equation}
which follows from ~\eqref{eq:cmi_fact_key} and the invariance of the TEE under topology-preserving deformations. Specifically, note the following identity
\begin{equation}
\begin{aligned} 
&I(\overline{u}:PQ{\setminus}(v\overline{u})|v)_{\sigma} - I(\overline{u}:Q{\setminus}(v\overline{u})|v)_{\sigma} \\
&= I(P:Q_2|Q_1)_{\sigma} - I(P:Q_2\setminus \overline{u}|Q_1)_{\sigma} \\
&=0.
\end{aligned}
\label{eq:explanation_for_e3}
\end{equation}
The equivalence of the first and the second line follows from the definition of the conditional mutual information. More specifically, note the following more explicit expressions for the two conditional mutual information quantities in the first line:
\begin{equation}
\begin{aligned}
&I(\overline{u}:PQ{\setminus}(v\overline{u})|v)_{\sigma} \\
&= S(\sigma_{\overline{u}v}) + S(\sigma_{PQ{\setminus}\overline{u}}) - S(\sigma_v) - S(\sigma_{PQ})
\end{aligned}
\end{equation}
and 
\begin{equation}
\begin{aligned}
&I(\overline{u}:Q{\setminus}(v\overline{u})|v)_{\sigma} \\
&= S(\sigma_{\overline{u}v}) + S(\sigma_{Q{\setminus}\overline{u}}) - S(\sigma_v) - S(\sigma_{Q}).
\end{aligned}
\end{equation}
Taking the difference between the two, we obtain
\begin{equation}
\begin{aligned}
&I(\overline{u}:PQ{\setminus}(v\overline{u})|v)_{\sigma} - I(\overline{u}:Q{\setminus}(v\overline{u})|v)_{\sigma} \\
&= \left(S(\sigma_Q) - S(\sigma_{PQ})\right) -  \left(S(\sigma_{Q{\setminus}\overline{u}}) - S(\sigma_{PQ{\setminus}\overline{u}}) \right).
\end{aligned}
\end{equation}
By adding $S(\sigma_{PQ_1}) - S(\sigma_{Q_1})$ to each paranthesis (so that they cancel out each other),  we obtain the second line of~\eqref{eq:explanation_for_e3}. The last line in \eqref{eq:explanation_for_e3} follows from the fact that the TEE is a constant. The second identity in~\eqref{eq:fact_third} is a direct consequence of ~\eqref{eq:cmi_reference}. This completes the proof.

\section{Canonical nature of max-entropy state} \label{appendix:2D}
This Appendix establishes further properties of the max-entropy state, but we emphasize our results here are not necessary for the argument in the main text nor in the other appendices.

The motivating question is as follows. Consider two topologically equivalent partitions of a chosen annulus $X$ as $X=ABC$ and $X=A'B'C'$ as in Fig.~\ref{fig:tee-spiral}. Is it true that the max-entropy state ($\lambda_X$) consistent with the reference state $\sigma$ on $AB$ and $BC$ must be the same as the max-entropy state ($\lambda'_X$) consistent with the reference state $\sigma$ on $A'B'$ and $B'C'$? We argue that the answer is yes for an appropriate notion of topological equivalence that we explain below. In particular, $A'B'C'$ need \emph{not} be obtainable from $ABC$ by using a quasi-1D approach (that is, by refining the partition only in the ``angular direction" and regrouping sites); see the partitions in Fig.~\ref{fig:tee-spiral} for an example of partitions that are not related in the quasi-1D way.

\begin{figure}[t]
    \centering
    	\begin{tikzpicture}
		\begin{scope}[scale=0.6]
			\draw [] (0,0) circle [radius=1];
			\draw [] (0,0) circle [radius=2];
			\draw (45:1) -- (45:2);
			\draw (135:1) -- (135:2);
			
			\draw (-45:1) -- (-45:2);
			\draw (-135:1) -- (-135:2);
			
			\node[] (P) at (0: 1.5)  {\small{$C$}};
			\node[] (P) at (-90: 1.5)  {\small{$B$}};
			\node[] (P) at (90: 1.5)  {\small{$B$}};
			\node[] (P) at (180: 1.5)  {\small{$A$}};
			
			\node[] (P) at (-90: 2.55)  {\small{(a)}};
			
		\end{scope}
		\begin{scope}[xshift=3.8 cm, scale=0.6]
			\draw [] (0,0) circle [radius=1];
			\draw [] (0,0) circle [radius=2];
			\draw [samples=100,smooth,domain=6.28:0,rotate around={45:(0,0)} ] 
			plot ({\x r}:{1+\x/6.28});   
			 
			\draw [samples=100,smooth,domain=6.28:0,rotate around={45+90:(0,0)} ] 
			plot ({\x r}:{1+\x/6.28});   
			
			\draw [samples=100,smooth,domain=6.28:0,rotate around={-45:(0,0)} ] 
			plot ({\x r}:{1+\x/6.28});   
 
			\draw [samples=100,smooth,domain=6.28:0,rotate around={-45-90:(0,0)} ] 
			plot ({\x r}:{1+\x/6.28});
			 
			\node[] (P) at (0: 2.5)  {\small{$C'$}};			 
			\node[] (P) at (90: 0)  {\small{$B'$}};
			\node[] (P) at (180: 2.5)  {\small{$A'$}};
			
		    \node[] (P) at (-90: 2.55)  {\small{(b)}};
			
			\draw [color=blue!60!white] (90:0.4) -- (90:1.07);
			\draw [color=blue!60!white] (-90:0.4) -- (-90:1.07);
			\draw [color=blue!60!white] (0:1.95) -- (0:2.3);
			\draw [color=blue!60!white] (180:1.95) -- (180:2.3);			
		\end{scope}	
	\end{tikzpicture}
    \caption{Topologically equivalent but geometrically distinct partitions of annulus $X$. (a) $X$ is $ABC$. (b) $X$ is $A'B'C'$.}
    \label{fig:tee-spiral}
\end{figure}
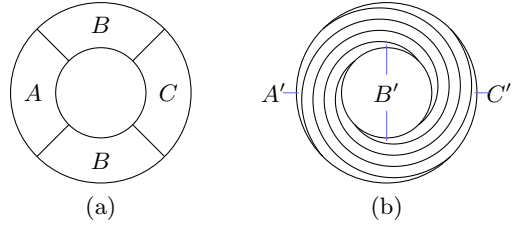

The key technical result is the following lemma about partitions of an annulus $X$ into five ordered regions and two moves in Fig.~\ref{fig:deformable-Xs}. Note that the combination of these two moves can ``deform" every region $X_i$ ($i=1,\cdots, 5$) in the partition. 

\begin{lemma}
\label{prop:canonical-lambda}
Let $X$ be an annulus and $\sigma$ be a reference state (Definition~\ref{assumption:TEE-inv}~of the main text). 
Let $\{X_1, X_2, X_3, X_4, X_5 \}$ be an ordered set of subsystems topologically equivalent to that shown in Fig.~\ref{fig:deformable-Xs}.
If $\{X_1, X_2, X_3, X_4, X_5 \}$ can be converted to a new partition $\{Y_1, Y_2, Y_3, Y_4, Y_5 \}$ of $X$ by a finite sequence of the two moves shown in Fig.~\ref{fig:deformable-Xs}, then $\tau_X(\{X_i\}) = \tau_X(\{Y_i\})$ where $\tau_X(\{X_i\})$ and $\tau_X(\{Y_i\})$ are the canonical Markov chains (Definition~\ref{definition:canonical_markov_chain}) with respect to the partitions $\{ X_1, \cdots, X_5\}$ and $\{ Y_1, \cdots, Y_5\}$ respectively.  
\end{lemma}

\begin{figure}[tb]
    \centering    
        \begin{tikzpicture}
        \begin{scope}[scale=0.9]
    \begin{scope}[scale=0.7]
    \draw[]  (0,0) rectangle (5,3);
    \draw[]  (1,1) rectangle (4,2);
    \draw[] (1,0) -- ++ (0,1);
    \draw[] (2,0) -- ++ (0,1);
    \draw[] (3,0) -- ++ (0,1);
    \draw[] (4,0) -- ++ (0,1);
    \draw[] (1,2) -- ++ (0,1);
    \draw[] (2,2) -- ++ (0,1);
    \draw[] (3,2) -- ++ (0,1);
    \draw[] (4,2) -- ++ (0,1);
        \node[] () at (1.5, 0.5) {$X_2$};
        \node[] () at (2.5, 0.5) {$X_3$};
        \node[] () at (3.5, 0.5) {$X_4$};
        \node[] () at (1.5, 2.5) {$X_2$};
        \node[] () at (2.5, 2.5) {$X_3$};
        \node[] () at (3.5, 2.5) {$X_4$};
        \node[] () at (4.5, 1.5) {$X_5$};
        \node[] () at (0.5, 1.5) {$X_1$};

         \draw [->, blue!60!black!80!white, line width=0.45 mm](1.5,0-0.25) --++ (-1.0,-1.0);
    \draw[->, green!40!black!80!white, line width=0.45 mm] (3.5,0-0.25) --++ (1.0,-1.0);
    \end{scope}

     \begin{scope}[scale=0.7,yshift=-4.5 cm, xshift=-3 cm]
    \draw[]  (0,0) rectangle (5,3);
    \draw[]  (1,1) rectangle (4,2);
    \draw[] (1-0.25,0) -- ++ (0.5,1);
    \draw[] (2,0) -- ++ (0,1);
    \draw[] (3,0) -- ++ (0,1);
    \draw[] (4-0.25,0) -- ++ (0.5,1)--++(-0.25,0.5);
    
    \draw[] (1, 1.5) --++(-0.25,0.5) -- ++ (0.5,1);
    \draw[] (2,2) -- ++ (0,1);
    \draw[] (3,2) -- ++ (0,1);
    \draw[] (4-0.25,2) -- ++ (0.5,1);
        \node[] () at (1.5, 0.5) {\color{blue!60!black!80!white}$X'_2$};
        \node[] () at (2.5, 0.5) {$X_3$};
        \node[] () at (3.5, 0.5) {\color{blue!60!black!80!white}$X'_4$};
        \node[] () at (1.5, 2.5) {\color{blue!60!black!80!white}$X'_2$};
        \node[] () at (2.5, 2.5) {$X_3$};
        \node[] () at (3.5, 2.5) {\color{blue!60!black!80!white}$X'_4$};
        \node[] () at (4.5, 1.5) {\color{blue!60!black!80!white}$X'_5$};
        \node[] () at (0.5, 1.5) {\color{blue!60!black!80!white}$X'_1$};
    \end{scope}
    
     \begin{scope}[scale=0.7,yshift=-4.5 cm, xshift=3 cm]
    \draw[]  (0,0) rectangle (5,3);
    \draw[]  (1,1) rectangle (4,2);
    \draw[] (1,0) -- ++ (0,1);
    \draw[] (2-0.25,0) -- ++ (0.5,1);
    \draw[] (3-0.25,0) -- ++ (0.5,1);
    \draw[] (4,0) -- ++ (0,1);
    \draw[] (1,2) -- ++ (0,1);
    \draw[] (2-0.25,2) -- ++ (0.5,1);
    \draw[] (3-0.25,2) -- ++ (0.5,1);
    \draw[] (4,2) -- ++ (0,1);
        \node[] () at (1.5, 0.5) {\color{green!40!black}$X'_2$};
        \node[] () at (2.5, 0.5) {\color{green!40!black}$X'_3$};
        \node[] () at (3.5, 0.5) {\color{green!40!black}$X'_4$};
        \node[] () at (1.5, 2.5) {\color{green!40!black}$X'_2$};
        \node[] () at (2.5, 2.5) {\color{green!40!black}$X'_3$};
        \node[] () at (3.5, 2.5) {\color{green!40!black}$X'_4$};
        \node[] () at (4.5, 1.5) {$X_5$};
        \node[] () at (0.5, 1.5) {$X_1$};
    \end{scope}
    \end{scope}
    \end{tikzpicture}
    \caption{A partition of annulus $X$ into ordered regions $\{ X_1, X_2, X_3, X_4, X_5\}$ and two moves preserving the topology. The first move gives $\{ X'_1, X'_2, X_3, X'_4, X'_5\}$, where $X'_1 X_2'=X_1 X_2$ and $X'_4 X_5'=X_4 X_5$. The second move gives $\{ X_1, X'_2, X'_3, X'_4, X_5\}$ such that $X'_2 X'_3 X'_4 = X_2 X_3 X_4$.}
    \label{fig:deformable-Xs}
\end{figure}
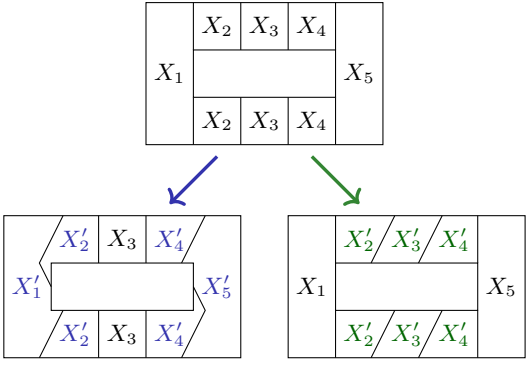

\begin{proof}
From the definition of $\sigma$ (Definition~\ref{assumption:TEE-inv}~in the main text), it follows that $\sigma$ has a constant $I(A:C|B)_{\sigma}$ for any chain-like $A, B,$ and $C$ that partition $X = \cup_{i=1}^5 X_i$. By Proposition~\ref{proposition:lqmc_characterization}, $\sigma_X$ is a locally Markov chain. Therefore, a canonical Markov chain associated with $\sigma_X$ and $\{X_1, X_2, X_3, X_4, X_5 \}$ exists, which we refer to as $\tau_X(\{X_i\})$. 

Next, we show that the two moves in Fig.~\ref{fig:deformable-Xs} do not change the canonical quantum Markov chain. The first move ``deforms" $\{X_1, X_2, X_3, X_4, X_5 \}$ to $\{X'_1, X'_2, X_3, X'_4, X'_5 \}$ such that $X_1 X_2= X_1' X_2'$ and $X_4 X_5= X_4' X_5'$ and such that the topology of the two sets is identical. (Note that $X_3$ is preserved.) We claim that $\tau_X(\{X_i\})=\tau_X(\{X_i'\})$. This is because $\tau_X(\{X_i\})$ is the max-entropy state consistent with $\sigma$ on $X_1X_2 X_3$ and $X_3X_4X_5$ (by Proposition~\ref{proposition:max_entropy}), and thus also the max-entropy state consistent with $\sigma$ on $X_1'X_2' X_3=X_1X_2 X_3$ and $X_3X_4' X_5'=X_3X_4 X_5$. 

The same idea works for the second move. The second move ``deforms"  $\{X_1, X_2, X_3, X_4, X_5 \}$ to $\{X_1, X'_2, X'_3, X'_4, X_5 \}$ such that $X_2 X_3 X_4 = X_2' X_3' X_4'$ and such that the topology of the two sets is identical. (Note that $X_1$ and $X_5$ are preserved.) The state $\tau_X$ is the max-entropy state consistent with $X_1X_2 X_3 X_4$ and $X_2 X_3 X_4 X_5$ (by Proposition~\ref{proposition:max_entropy}), and therefore it is also the max-entropy state consistent with $X_1X'_2 X'_3 X'_4=X_1X_2 X_3 X_4$ and $X'_2 X'_3 X'_4 X_5=X_2 X_3 X_4 X_5$. 

Thus, as long as an ordered partition of $X$ as $\{X_1,X_2,X_3,X_4,X_5\}$ can be converted to the ordered partition $\{Y_1,Y_2,Y_3,Y_4,Y_5\}$ by a finite sequence of the two moves depicted in Fig.~\ref{fig:deformable-Xs}, the canonical Markov chain associated with the former partition, $\tau_X(\{X_i\})$, must be identical with $\tau_X(\{Y_i\})$, the canonical Markov chain associated with the latter partition. This completes the proof.
\end{proof}

We can now use Lemma~\ref{prop:canonical-lambda} to argue that the max-entropy state ($\lambda_X$) consistent with the reference state on $AB$ and $BC$ in Fig.~\ref{fig:tee-spiral}(a) is identical to the max-entropy state ($\lambda'_X$) consistent on the ``spiral-shaped'' subsystems $A'B'$ and $B'C'$, shown in Fig.~\ref{fig:tee-spiral}(b). First, we partition $B$ further as $B=X_2 X_3 X_4$ and similarly, $B'=Y_2 Y_3 Y_4$, and let $A=X_1, C=X_5$, $A'=Y_1, C'=Y_5$. Suppose that the size (and thickness) of each subsystem is large compared to the lattice spacing, so that it is possible to deform $\{X_1,X_2,X_3,X_4,X_5\}$  to $\{ Y_1,Y_2,Y_3,Y_4,Y_5\}$ by a finite sequence of moves discussed in Lemma~\ref{prop:canonical-lambda}, and Fig.~\ref{fig:deformable-Xs}. Then $\tau_X(\{X_i\})=\tau_X(\{Y_i\})$ according to Lemma~\ref{prop:canonical-lambda}. By Proposition~\ref{proposition:max_entropy}, $\lambda_X=\tau_X(\{X_i\})$ and $\lambda'_X=\tau_X(\{Y_i\})$. Therefore, $\lambda_X = \lambda'_X$. This completes the argument.\footnote{In order for this argument to work, it is important to assume that the chosen subsystems are sufficiently large. Otherwise, it may not be possible to deform $\{ X_1,X_2,X_3,X_4,X_5\}$ to $\{Y_1,Y_2,Y_3,Y_4,Y_5 \}$ using the moves discussed in Lemma~\ref{prop:canonical-lambda}.}

This argument leads to the conclusion that the max-entropy state for a given annulus $X$ is \emph{canonical}. Insofar as two different partitions of an annulus are ``topologically equivalent'' (in the sense that the partitions can be deformed using a sequence of moves described in Lemma~\ref{prop:canonical-lambda}), the max-entropy states associated with those partitions are identical. This generalizes the observation that one can associate a canonical Markov chain to a one-dimensional chain (Appendix~\ref{sec:local_vs_global_qmc}) to two dimensions. 

\section{Heuristic argument for generalization beyond doubled phases} \label{sec:stacking}
Here we present a heuristic argument showing that the lower bound 
\begin{align} \label{eq:supp_main_bound}
    \gamma \geq \log \mathcal{D}
\end{align}
from the main text holds for any 2D gapped bosonic ground state. The key point is that reference states (or more specifically, string-net states) are believed to realize all ``doubled'' 2D topological phases, i.e., those obtained by stacking a bosonic topological phase onto its time-reversed partner. So one expects that for any 2D gapped bosonic ground state $\rho$, we should be able to write the doubled state $\rho \otimes \rho^*$ as
\begin{align}
\rho \otimes \rho^* = U \sigma U^\dagger
\end{align}
for some valid reference state $\sigma$ and some appropriate unitary $U$. Here $\rho^*$ denotes the complex conjugate of $\rho$ in some local product basis, while $\rho \otimes \rho^*$ denotes a bilayer state constructed from stacking $\rho^*$ onto $\rho$. One caveat here is that the unitary $U$ need not be a constant-depth circuit in general; instead we have to allow for the possibility that $U$ is a \emph{quasi-local} circuit, i.e.\ a circuit with decaying tails, since circuits of this kind arise naturally from the adiabatic flow of Hamiltonians \cite{Hastings2005, nachtergaele2019quasi}. In what follows, we ignore this issue and think of $U$ as a standard constant-depth circuit. In this case, it follows that the lower bound \eqref{eq:supp_main_bound} holds for the doubled state $\rho \otimes \rho^*$.  It then follows that the bound \eqref{eq:supp_main_bound} must also hold for $\rho$ itself, since the TEE of $\rho \otimes \rho^*$ is exactly twice that of $\rho$, as is the total quantum dimension $\log \mathcal{D}$. Since $\rho$ is an arbitrary gapped bosonic ground state, we have established the claim.
Of course, this argument is only heuristic; to turn this into a rigorous proof, we would need to (i) prove the existence of the quasi-local circuit $U$ and (ii) generalize the results in this paper from constant-depth circuits to quasi-local circuits. We leave these questions to future work.

\end{document}

%% file: fig_annulus_to_chain.tex
 \begin{tikzpicture}
  \centering
  \begin{scope} [scale=0.7]
    \begin{scope}
      \draw (0,0) rectangle (5,3);
      \draw (1,1) rectangle (4,2);
      \foreach \x in {1,2,3,4}{
        \draw (\x,0) -- (\x,1) (\x,2) -- (\x,3);
      }
  
      \draw (0.5,1.5)node{$X_1$};
      \draw (1.5,0.5)node{$X_2$};
      \draw (1.5,2.5)node{$X_2$};
      \draw (2.5,0.5)node{$X_3$};
      \draw (2.5,2.5)node{$X_3$};
      \draw (3.5,0.5)node{$X_4$};
      \draw (3.5,2.5)node{$X_4$};
      \draw (4.5,1.5)node{$X_5$};
    \end{scope}
  \end{scope}
\end{tikzpicture}